\documentclass[12pt]{article}
\usepackage{amsmath,amssymb,amsthm,booktabs}
\usepackage{mathrsfs}
\usepackage{geometry}
\usepackage{pifont}
\usepackage[OT1]{fontenc}
\usepackage[utf8]{inputenc}
\usepackage[colorlinks,citecolor=blue,urlcolor=blue]{hyperref}
\usepackage{txfonts}
\usepackage{indentfirst}
\usepackage{multirow}
\usepackage{float}

\usepackage{ragged2e}  
\usepackage{array}     
\usepackage{geometry}  
\usepackage{tikz}
\usetikzlibrary{arrows.meta, positioning}
\usepackage{graphicx}
\usepackage{subcaption}
\usepackage{algorithmic}
\usepackage{siunitx}    

\usepackage{bm}
\usepackage{euscript}
\usepackage{graphicx}
\usepackage{multicol}
\usepackage[round]{natbib}
\usepackage[ruled]{algorithm2e}
\usepackage{colortbl}

\numberwithin{equation}{section} \theoremstyle{plain}
\newtheorem{theorem}{Theorem}[section]
\newtheorem{lemma}{Lemma}[section]

\newtheorem{proposition}{Proposition}[section]

\newtheorem{definition}{Definition}[section]
\newtheorem{remark}{Remark}[section]

\begin{document}

\newcommand{\gai}[1]{{#1}}


\makeatletter
\def\ps@pprintTitle{%
  \let\@oddhead\@empty
  \let\@evenhead\@empty
  \let\@oddfoot\@empty
  \let\@evenfoot\@oddfoot
}
\makeatother

\newcommand\tabfig[1]{\vskip5mm \centerline{\textsc{Insert #1 around here}}  \vskip5mm}

\vskip2cm

\title{Pareto frontier of portfolio investment under volatility uncertainty and short-sale constraints market}
\author{Jing He\thanks{School of Mathematics, Shandong University, PR China, (hjing@mail.sdu.edu.cn).}
\quad Shuzhen Yang\thanks{Shandong University-Zhong Tai Securities Institute for Financial Studies, Shandong University, PR China, (yangsz@sdu.edu.cn). }
}

\date{}
\maketitle

\begin{abstract}
In this paper, we investigate a portfolio investment problem under volatility uncertainty and short-sale constraints market via sublinear expectation which is used to model volatility uncertainty. We assume the stocks admit volatility uncertainty. Thus the related portfolio has upper variance (maximum risk) and lower variance (minimum risk). 
By introducing a risk factor $w$ to conduct coupled modeling of the maximum and minimum risks, a simplified Sublinear Expectation Mean-Uncertainty Variance (SLE-MUV) model is constructed. 
Theoretically, we show that the Pareto frontier of the SLE-MUV model is a continuous convex curve, and its optimal solution can be expressed as a polynomial analytical expression with respect to the risk factor $w$. 
Empirically, we systematically test the practical performance of the SLE-MUV model and conduct comparative analysis with the traditional Mean-Variance (MV) model as the benchmark based on three sets of samples --- simulated generated data, data of the US stock market and the A-share market. The empirical results show that the SLE-MUV model can significantly improving the risk-adjusted return of the investment portfolio.
\end{abstract}

\newpage

\section{Introduction}
Well known that Brownian motion is always used to describe the return of stock under a given probability space. In practical analysis, it is not easy to estimate the variance of the return. Therefore, we consider to introduce the volatility uncertainty in the classical Brownian motion, where the volatility of the Brownian motion takes value in a interval $[\underline{\sigma}t,\overline{\sigma}t]$. In this paper, we will investigate the portfolio investment problem under volatility uncertainty Brownina motion. 

The very first mathematical investment model is Mean-Variance type, which is 
 proposed by \cite{Markowitz1952} optimizes the process of selecting investment portfolios, enabling investors to minimize risk at a given level of expected return or maximize expected return at a given level of risk and has been adopted in the vast majority of existing studies.
\cite{mandelbrot1963} challenged the Gaussian assumption in financial returns by showing that empirical price changes exhibit heavy tails. Such empirical evidence is fundamentally incompatible with the strong premises of the Mean-Variance framework --- namely, normally distributed returns or quadratic preferences. This inconsistency triggered an influential academic controversy regarding the limits of portfolio optimization theory \citep{samuelson1967,feldstein1969, borch1969,samuelson1970}. 
\cite{liu2003mean, yu2008neural,li2010mean} proposed a mean-variance-skewness model for portfolio selection, in which skewness, defined as the third-order moment of returns, is adopted to characterize the asymmetry of return distributions.

Financial markets are vulnerable to non-random shocks, so distributions derived from historical data often fail to fully reflect the actual distribution of future returns. Accordingly, increasing research attention has been paid to financial data with uncertainty. Uncertainty Theory proposed by \cite{liu2013uncertain} is employed to address portfolio optimization problems where historical data are insufficient to estimate the return distributions \citep{zhai2022uncertain,qin2015mean,belabbes2025uncertain}.
\cite{LI2026666} adopted Wasserstein-based distributionally robust optimization to tackle financial market uncertainty, proposing a distributionally robust optimal allocation strategy for financial assets.
\cite{CAI2024310} derived worst-case expressions for distorted risk measures of stopped-loss and bounded loss random variables, identifying distributions that achieve these worst-case values. 
\cite{CAI2025905} obtained closed-form expressions for the worst-case target semi-variance when only the mean and variance of losses are known and extended the mean-variance model to a mean-target semi-variance model. 

Most of data series such as stock returns are generally characterized by leptokurtosis, fat tails and heteroscedasticity, which renders their distribution patterns with notable uncertainty and complexity and makes it difficult to accurately depict them with a single distribution. To address this prevalent practical issue, \cite{peng2007,Peng2010} pioneeringly constructed the framework of the sublinear expectation space $(\Omega,\mathcal{H},\mathbb{E})$ and established a comprehensive theoretical framework that enables us to characterize and analyze distributionally uncertain data. 
This theory has demonstrated significant effectiveness in financial applications. \cite{PengYang2020} proposed a Value-at-Risk (VaR) prediction model based on sublinear expectation theory, namely the G-VaR model, and extensive empirical tests show that this model outperforms most existing benchmark VaR prediction models with excellent predictive performance. \cite{Peng2022} constructed an autoregressive model to calibrate the parameters of the G-normal distribution, and the performance is consistent with the findings of \cite{PengYang2020}.

Given the volatility (distributional) uncertainty of stock return, this paper introduces the sublinear expectation theory into the traditional single-period mean-variance model and derives a risk factor $w$ through theoretical analysis, that is 
$$
\min_{\sum_{i=1}^n\beta_i=1,\ \beta_i\geq 0}\   w\underline{\sigma}^2_{\beta}+(1-w) \overline{\sigma}^2_{\beta},
$$
with a lower bound constraint of the mean, where $\underline{\sigma}^2_{\beta}$ and $\overline{\sigma}^2_{\beta}$ are the minimum and maximum variance of the portfolio.  
Indeed, difficulties arise when estimating the upper and lower covariance matrices of asset returns. We find that this problem can be transformed into solving for the maximum and minimum values of each element in the covariance matrix, subject to the constraint \(\beta_i\geq 0\) for \(i=1,2,\cdots,n\), where the constraint implies a prohibition on short selling risky assets. Further details are provided in Subsection \ref{sub:mv_sle}.
Investors can select this risk factor $w$ according to their own investment preferences and the current market environment, and conduct portfolio allocation by virtue of the mean-variance model under the sublinear expectation. Comprehensive empirical tests are conducted to verify the effectiveness and flexibility of the proposed method. 

Methods of \cite{CAI2024310,CAI2025905}  and the proposed method in this paper are grounded in the core assumption of distribution uncertainty in losses/returns. The key distinction lies in the way uncertainty is modeled:
The distributional uncertainty framework adopted by \cite{CAI2024310,CAI2025905} is based on fixed and known mean and variance of the underlying distribution. By constructing a Wasserstein ambiguity set or reducing the problem to finite-point distributions, they solve for the worst-case risk over all feasible distributions satisfying the moment constraints.
In contrast, the proposed method in this paper characterizes volatility uncertainty using the G-normal distribution, under which the variance is treated as uncertain parameters. Using the $\varphi$-max-mean algorithm, we derive the upper and lower bounds of the variance, thereby enabling the accurate quantification of the maximum and minimum risks. 

The core contributions of this paper are threefold, as outlined below:

\textbf{Theoretical Framework Extension}. We advance the classic Mean-Variance (MV) portfolio model by integrating sublinear expectation theory into its standard formulation. Departing from the conventional single-objective optimization paradigm, we explicitly account for both the upper and lower variances of portfolio returns. 

\textbf{Rigorous Theoretical Foundation}. To resolve the solvability issue of the multi-objective (upper variance and lower variance) optimization problem, we introduce a risk factor $w \in [0,1]$, which converts the multi-objective formulation into a tractable single-objective optimization model, namely the SLE-MUV model. We theoretically demonstrate that the Pareto frontier of the SLE-MUV model is a continuous convex curve and derive the closed-form optimal solution of the proposed model using the active set method.

\textbf{Empirical Validation}.
We conduct extensive empirical analyses to evaluate the performance of the SLE-MUV model across three distinct datasets: synthetic data, six representative U.S. stocks, and six representative Chinese A-shares, with a head-to-head benchmark comparison against the traditional MV model. The empirical results consistently show that the proposed model offers high flexibility.

The structure of this paper is organized as follows:
Section \ref{sec:SLE-MUV model} introduces the fundamental theories of the traditional Mean-Variance (MV) model, and extends it to the framework of sublinear expectation theory to construct the Mean-Variance model under sublinear expectation (SLE-MUV).
Section \ref{sec:Pareto frontier and Optimal Solution} focuses on the Pareto frontier analysis of the SLE-MUV model, and derives its theoretical optimal solution using the active set method.
Subsequently, we validate the performance of the proposed SLE-MUV model using three sets of empirical data: synthetic data, six U.S. stocks, and six A-shares, and compare it with the traditional Mean-Variance model in Section \ref{sec:Simulation}.
Finally, Section \ref{sec:conclusion} summarizes the key findings and main contributions of this paper.

\section{Mean-Variance model}\label{sec:SLE-MUV model}
The traditional single-period Mean-Variance model assumes a frictionless market with no taxes or transaction costs and prohibits short selling. Investors assess portfolio risk via variance or standard deviation and judge performance by expected returns. As rational and risk-averse individuals, they make investment decisions based only on risk and return. Specifically, they choose lower-risk portfolios for equal returns and higher-return portfolios at the same risk level.

Suppose that an investor selects $n$ assets, where the portfolio allocation vector is denoted by $ \bm{\beta} = (\beta_1, \beta_2, \ldots, \beta_n)^T $, $\sum_{i=1}^{n}\beta_i = 1$, $\beta_i \ge 0, i =1,2, \dots ,n$, and the vector of asset returns is given by $ \bm{\mu} = (\mu_1, \mu_2, \ldots,\mu_n)^T$.  
The covariance matrix of asset returns is $V$. The minimizing expected portfolio return is $r_0$ and the maximizing expected portfolio risk is $\sigma_0^2$.  
By balancing the two indicators of maximizing returns or minimizing risks, the classical mean-Variance model is 
\begin{equation}
\begin{array}{ll}
\min & \bm{\beta}^\top V \bm{\beta}, \\
\text{s.t.} & \bm{\beta}^\top \bm{\mu} \geq r_0, \\
& \sum_{i=1}^n \beta_i = 1, \\
& \beta_i \geq 0, i=1,2,\cdots,n,
\end{array}
\quad \text{or} \quad
\begin{array}{ll}
\max & \bm{\beta}^\top \bm{\mu}, \\
\text{s.t.} & \beta^\top V \beta \leq \sigma_0^2, \\
& \sum_{i=1}^n \beta = 1, \\
& \beta_i \geq 0, i=1,2,\cdots,n.
\end{array}
\end{equation}

However, asset returns exhibit significant distributional uncertainty, influenced by a variety of complex factors. 
This remains a persistent challenge to accurately characterize the actual return and risk characteristics of investment portfolios. 

\subsection{Mean-Variance model under SLE}\label{sub:mv_sle}
This study modifies the Mean-Variance model by incorporating Sublinear Expectation theory, 
thereby deriving a portfolio risk interval for investors and facilitating the determination of optimal asset allocation ratios. We now regard any asset which is a financial time-series $\{X_{t}\}_{0\le t\le T}$, governed not by a single distribution but by an infinite family of distributions $\{F_{\theta}\}_{\theta \in \Theta}$, each capturing specific properties of the data \citep{PengYang2020}. 
Consider a relatively simple scenario --- the return of the asset is constant, but the variance varies over time within a interval.

Suppose there are $n$ assets returns $\{X_{t}^{(i)}\}_{0\le t\le T}$, $i = 1,2,...,n$, each following a G-normal distribution $\mathcal{N}(\mu_i,[\underline{\sigma}_i^2,\overline{\sigma}_i^2])$. Let $V$ denote the covariance matrix of these $N$ asset returns. It can be derived that the element $ V_{ij} = \sigma^2_{ij} $ of matrix $V$ lies within the interval $[\underline{\sigma}^2_{ij},\overline{\sigma}^2_{ij}]$ for all $ i, j = 1, 2, \ldots, n $, where
 \begin{align*}
 \overline{\sigma}_{ij}^2 & = \mathbb{E}[(X_i-\mathbb{E}[X_i])(X_j-\mathbb{E}[X_j])] = \mathbb{E}[X_iX_j]-\mu_i\mu_j, \\
   \underline{\sigma}_{ij}^2 & = -\mathbb{E}[-(X_i-\mathbb{E}[X_i])(X_j-\mathbb{E}[X_j])] = -\mathbb{E}[-X_iX_j]-\mu_i\mu_j.
 \end{align*}
Let $\underline{V}$ and $\overline{V}$ be two $n \times n$ matrices, where $\underline{V}_{ij} = \underline{\sigma}^2_{ij}$, $\overline{V}_{ij} = \overline{\sigma}^2_{ij}$.
We further assume that $V$, $\underline{V}$ and $\overline{V}$ are all positive definite matrices.
Our result reduces to the one established in \cite{Pei2025} under the condition $\mathbb{E}\left[X_{t}^{(i)}\right]=0$, $i = 1,2,...,n$.

For each fixed portfolio allocation $\bm{\beta}$ of n assets,
the expected return of the portfolio is $\mu_{\beta} = \bm{\beta}^{\top}\bm{\mu}$,
the variance is given by $\sigma_\beta^2 = \bm{\beta}^T V \bm{\beta} \in [\underline{\sigma}_\beta^2, \overline{\sigma}_\beta^2]$, which lies within the corresponding uncertainty interval owing to the ambiguity in the covariance matrix.  
Since $\beta_i \ge 0, i =1,2, \dots ,n$, we have that
\begin{align}
\overline{\sigma}^2_{\beta} &= \max \bm{\beta}^{\top}V\bm{\beta}  =
\bm{\beta}^{\top}\overline{V}\bm{\beta} =\sum_{i=1}^n\sum_{j=1}^n \beta_i\beta_j\overline{\sigma}^2_{ij},\label{eq:overline_sigma}\\
\underline{\sigma}^2_{\beta}  &=\min \bm{\beta}^{\top}V\bm{\beta}  = \bm{\beta}^{\top}\underline{V}\bm{\beta} = \sum_{i=1}^n\sum_{j=1}^n \beta_i\beta_j\underline{\sigma}^2_{ij}.\label{eq:underline_sigma}
\end{align}

Variance is a commonly used indicator to measure the volatility and risk level of asset returns. 
In the case of portfolio distributional uncertainty, the variance of the portfolio is no longer uniquely determined, but is within an upper and lower bound interval ($\sigma_\beta^2 \in [\underline{\sigma}_\beta^2, \overline{\sigma}_\beta^2]$). $\overline{\sigma}_\beta^2$ and $\underline{\sigma}_{\beta}^2$ represent the maximum and minimum potential risk respectively.
Therefore, the risk of an investment portfolio cannot be precisely quantified by a single value, but can be controlled through upper and lower variance intervals. 

According to modern portfolio theory, rational investors aim to maximize returns while minimizing risks. After incorporating the sublinear expectation theory, regardless of their risk preferences, rational investors should take both upper and lower variances into account to ensure that risks remain within a controllable range. On this basis, a multi-objective mean-variance optimization model based on the sublinear expectation theory can be constructed.

\begin{align}\label{eq:MO-MV eq}
\min& \ \underline{\sigma}^2_{\beta},\ \min \overline{\sigma}^2_{\beta},\\
s.t. & \sum_{i=1}^{n}\beta_i = 1,\nonumber \\
&\   \bm{\beta}^{\top}\bm{\mu} \ge r_0, \nonumber\\
&\  \beta_i \ge 0,i=1,2,\cdots,n.\nonumber
\end{align}

We further assume that the actual variance on a certain day can be approximated by a linear combination of the upper and lower variances, i.e. there exists a risk factor $w \in [0,1]$ satisfying $\hat{\sigma}_{\beta}^2 = w\underline{\sigma}_{\beta}^2 + (1-w)\overline{\sigma}_{\beta}^2$, then the multi-objective model can be transformed into a single objective form Eq. \eqref{eq:SP-MV eq}. 
By using the linear weighted sum method, we can also obtain the following single objective optimization problem from an algorithmic perspective, which is consistent with the actual situation. The subsequent section \ref{sec:Pareto frontier and Optimal Solution} is devoted to analyzing the relationship between the optimal solutions corresponding to Eq. \eqref{eq:MO-MV eq} and \eqref{eq:SP-MV eq}.
\begin{align}\label{eq:SP-MV eq}
   \min&\   w\underline{\sigma}^2_{\beta}+(1-w) \overline{\sigma}^2_{\beta}, \\
   s.t. & \sum_{i=1}^{n}\beta_i = 1,\nonumber \\
   &\   \bm{\beta}^{\top}\bm{\mu} \ge r_0, \nonumber\\
   &\  \beta_i \ge 0, i =1,2,\cdots,n.\nonumber
\end{align}

When the weighted sum approach is employed to address multi-objective optimization problems, the weight (risk factor $w$) assigned to the upper and lower bounds of the variance reflect investor's risk preferences. 
For the resulting single-objective function, the risk factor $w$ captures an investor's confidence in the lower bound of portfolio risk. A larger value of $w$ indicates that the investor is more confident that the portfolio risk is relatively low and the market is more stable. Conversely, a smaller value of $w$ implies that the investor tends to perceive a higher portfolio risk level and a more volatile market environment.
Investors are required to  determine the value of $w$ based on the actual performance of the investment portfolio. In Section \ref{sec:Simulation}, we conduct empirical experiments using simulated data, U.S. stock market data, and A-share market data as samples respectively. Based on the actual return rates and Sharpe ratios of each portfolio, we will conduct specific analysis and discussion on the reasonable value of $w$.

\section{Pareto frontier analysis for SLE-MUV Model}\label{sec:Pareto frontier and Optimal Solution}
In this section, we conduct a Pareto frontier analysis and derive the solution to the mean-variance model under the sublinear expectation theory, so as to further determine the importance of the risk factor \( w \).
\subsection{Pareto frontier Analysis}
In this subsection, we consider the convex multi-objective portfolio optimization model based on upper and lower variance risk measures. 
By drawing on conclusions of \citep{deb2016,He2015,Dylan2022}, we can conduct a Pareto frontier analysis of the mean-variance model under the sublinear expectation framework, thereby providing theoretical support for the optimal decision-making of such models.
\begin{lemma}\label{Le:Lemma1}
For each given $w \in (0,1)$ (or $w \in [0,1] $), the corresponding optimal solution of Eq.  \eqref{eq:MO-MV eq} must be a effective solution (or weakly effective solution) of Eq. \eqref{eq:SP-MV eq}.
\end{lemma}
\begin{proof}
  Let $\bm{\beta}^*$ is the optimal solution of Eq. \eqref{eq:MO-MV eq}, which is not a effective solution of Eq. \eqref{eq:SP-MV eq}, there exists a $\overline{\bm{\beta}} \in X_{\boldsymbol{\beta}} =\left\{\bm{\beta} \in \mathbb{R}^n| \sum_{i=1}^{n}\beta_i = 1, \sum_{i=1}^{n}\beta_i \mu_i \ge \mu, \beta_i \ge 0\right\}$, s.t.
  \begin{align*}
    \overline{\sigma}^2(\overline{\bm{\beta}}) \le \overline{\sigma}^2(\bm{\beta}^*),\quad
    \underline{\sigma}^2(\overline{\bm{\beta}})\le \underline{\sigma}^2(\bm{\beta}^*). 
  \end{align*}
  Hence we have $w\in (0,1)$,
  \begin{align*}
    w\underline{\sigma}^2(\overline{\bm{\beta}}) +(1-w)\overline{\sigma}^2(\overline{\bm{\beta}}) \le w\underline{\sigma}^2(\bm{\beta}^*) + (1-w)\overline{\sigma}^2(\bm{\beta}^*),
  \end{align*}
  This contradicts $\bm{\beta}^*$ as the optimal solution of Eq. \eqref{eq:MO-MV eq}. The proof for the scenario with $w \in [0,1]$ follows analogously.
\end{proof}
\begin{lemma}\label{Le:Lemma2}
	When $\underline{V}$, $\overline{V}$ are positive definite matrices,
    for any weakly effective solution (or effective solution) $\bm{\beta}^*$ of Eq. \eqref{eq:MO-MV eq}, there must exist a $w \in [0,1]$ (or $w \in (0,1)$) such that $\bm{\beta}^*$ is the optimal solution of Eq. \eqref{eq:SP-MV eq}.
\end{lemma}
\begin{proof}
  Let $\bm{\beta}^*$ is a weakly effective solution of Eq. \eqref{eq:MO-MV eq}, there exists $\bm{\beta} \in X_{\boldsymbol{\beta}} =\left\{ \bm{\beta} \in \mathbb{R}^n| \sum_{i=1}^{n}\beta_i = 1 \right.$, $ \left. \sum_{i=1}^{n}\beta_i \mu_i \ge \mu, \beta_i \ge 0\right\}$, s.t.
  \begin{align*}
        \overline{\sigma}^2(\overline{\bm{\beta}}) < \overline{\sigma}^2(\bm{\beta}^*),\quad
    \underline{\sigma}^2(\overline{\bm{\beta}})< \underline{\sigma}^2(\bm{\beta}^*) .
  \end{align*}
  	When  $\underline{V}$, $\overline{V}$ are positive definite matrices, $\underline{\sigma}_{\beta}^2$ and $\overline{\sigma}_{\beta}^2$ are convex functions, there exist two scalars $w_1 ,w_2 \ge 0$  that are not both zero, s.t. for all $\beta \in X_{\boldsymbol{\beta}}$,
  \begin{equation*}
    w_1\left( \underline{\sigma}^2(\bm{\beta}) - \underline{\sigma}^2(\bm{\beta}^*) \right) + w_2 \left( \overline{\sigma}^2(\bm{\beta}) -\overline{\sigma}^2(\bm{\beta}^*) \right) \ge 0,
  \end{equation*}
  i.e.
  \begin{equation*}
    w_1\underline{\sigma}^2(\bm{\beta}) + w_2\overline{\sigma}^2(\bm{\beta}) \ge w_1\underline{\sigma}^2(\bm{\beta}^*) + w_2\overline{\sigma}^2(\bm{\beta}^*).
  \end{equation*}
  Without loss of generality, assume $w_1 + w_2 = 1$, i.e., $w_1 \in [0,1]$ and $w_2 = 1 - w_1$, hence $\bm{\beta}^*$ is the optimal solution of Eq. \eqref{eq:SP-MV eq}.
\end{proof}
\begin{remark}
The above two lemmas indicate that for the convex multi-objective programming problem Eq. \eqref{eq:MO-MV eq}, when $w$ takes all values in $[0,1]$, all weakly efficient solutions of Eq. \eqref{eq:MO-MV eq} can be obtained by solving the optimal solutions of Eq. \eqref{eq:SP-MV eq}. However, when $w$ takes all values in $(0,1)$, although efficient solutions of Eq. \eqref{eq:MO-MV eq} can also be derived via Eq. \eqref{eq:SP-MV eq}, it is not guaranteed that all efficient solutions can be obtained.
\end{remark}

	\begin{proposition}\label{property: smooth convex curve}
		When $\underline{V}$, $\overline{V}$ are positive definite matrices, the Pareto frontier of the convex multi-objective programming problem Eq. \eqref{eq:MO-MV eq} is a continuous convex curve.
	\end{proposition}
\begin{proof}
    It is obvious that the problem given by Eq. \eqref{eq:SP-MV eq} is a convex optimization problem with a strictly convex objective function when $w \in [0,1]$. Therefore, Eq. \eqref{eq:SP-MV eq} has a unique optimal solution, which implies that Eq. \eqref{eq:MO-MV eq} has a unique weakly efficient solution by Lemma \ref{Le:Lemma1} and Lemma \ref{Le:Lemma2}.
    
   Let $f(\bm{\beta},w) = w\underline{\sigma}^2(\bm{\beta})+(1-w)\overline{\sigma}^2(\bm{\beta})$, define optimal value function $V(w) = min_{\bm{\beta}\in F}f(\bm{\beta},w)$, where the feasible set $F$ is compact, $f(\bm{\beta},w)$ is continuous with respect to 
   $(\bm{\beta},w)$ and strictly convex with respect to $\beta$.
   By Berge's Maximum Theorem, the optimal solution set mapping 
   $argmin_{\bm{\beta}\in F}f(\bm{\beta},w)$ is upper hemicontinuous. Since the optimal solution is unique, upper hemicontinuity is equivalent to continuity, $\bm{\beta}^{*}(w)$ is continuous and the parametric representation of the Pareto frontier $\gamma(w) = \left(\underline{\sigma}^2_{\bm{\beta}(w)},\overline{\sigma}^2_{\bm{\beta}(w)}\right)$ is a continuous curve for 
   $w\in[0,1]$.
   
   Let 
   $\gamma(w_1) = \left(\underline{\sigma}^2_{\bm{\beta}(w_1)},\overline{\sigma}^2_{\bm{\beta}(w_1)}\right)$  and 
   $\gamma(w_2) = \left(\underline{\sigma}^2_{\bm{\beta}(w_2)},\overline{\sigma}^2_{\bm{\beta}(w_2)}\right)$  be any two points on the Pareto frontier, corresponding to the optimal solutions $\bm{\beta}(w_1)$ and $\bm{\beta}(w_2)$ in the feasible region, respectively.
   For any $\lambda \in [0,1]$, consider the convex combination in the feasible region:
   $\bm{\beta}_{\lambda} = \lambda\bm{\beta}(w_1) + (1-\lambda)\bm{\beta}(w_2)$.
   Since the feasible region is a convex set,  $\bm{\beta}_\lambda $  remains within the feasible region. Due to convex objective function, we have
   \begin{align*}
   	\underline{\sigma}_{\bm{\beta}_{\lambda}}^2 &\le \lambda\underline{\sigma}^2_{\bm{\beta}(w_1)} + (1-\lambda)\underline{\sigma}^2_{\bm{\beta}(w_2)},\\
   	\overline{\sigma}_{\bm{\beta}_{\lambda}}^2 &\le \lambda\overline{\sigma}^2_{\bm{\beta}(w_1)} + (1-\lambda)\overline{\sigma}^2_{\bm{\beta}(w_2)}.
   \end{align*}
   The Pareto front dominates or lies on $\left(\lambda\underline{\sigma}^2_{\bm{\beta}(w_1)} + (1-\lambda)\underline{\sigma}^2_{\bm{\beta}(w_2)}, \lambda\overline{\sigma}^2_{\bm{\beta}(w_1)} + (1-\lambda)\overline{\sigma}^2_{\bm{\beta}(w_2)}\right)$, implying that the Pareto front is convex.
   
   Hence for all $ w \in [0,1] $, the Pareto frontier of Eq. \eqref{eq:MO-MV eq} is a unique continuous convex curve.
\end{proof}
\begin{figure}
  \centering
  \includegraphics[width=0.6\textwidth]{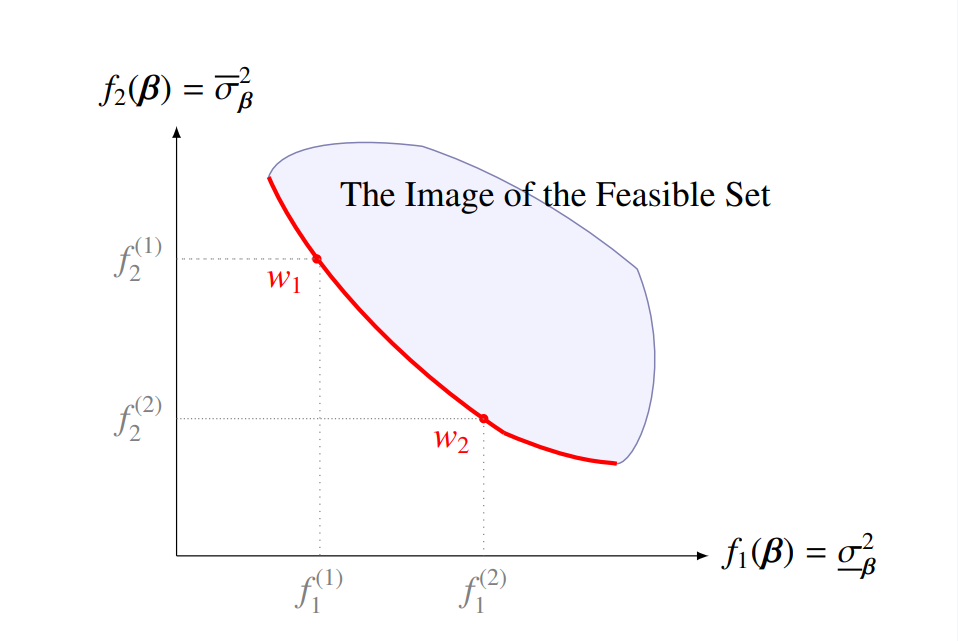}
	\caption{Geometric illustration of a smooth convex Pareto frontier
		between the upper and lower variance objectives.}
	\label{fig:pareto_frontier_smooth}
\end{figure}

In Figure \ref{fig:pareto_frontier_smooth}, the blue region illustrates the image of the feasible set 
$X_{\boldsymbol{\beta}} =\left\{ \bm{\beta} \in \mathbb{R}^n| \sum_{i=1}^{n}\beta_i = 1,\right.$ $\left. \bm{\beta}^{\top}\bm{\mu} \ge r_0, \beta_i \ge 0, i = 1,2,\cdots,n \right\}$ under the objective functions, which is convex. The red curve indicates the Pareto frontier, representing the convex boundary of efficient solutions. The labeled points on the curve correspond to the optimal solutions obtained for different risk factors $w$.

\subsection{Optimal Solution for SLE-MUV Model}\label{sec:The role of weights in portfolio allocation}
In this subsection, we derive the analytical expression for the optimal solution of Eq. \eqref{eq:SP-MV eq} by means of the active set method, and further obtain the unique weakly efficient solution of Eq. \eqref{eq:MO-MV eq}.

Eq. \eqref{eq:SP-MV eq} can be simplified to $\bm{\beta}^{\top}\Sigma \bm{\beta}$, where $\Sigma  = w \underline{V}+(1-w)\overline{V}$. We now construct the Lagrangian function,
\begin{align}\label{eq:lagrangian function}
  \mathcal{L}(\bm{\beta}) = \bm{\beta}^{\top} \Sigma \bm{\beta} + \lambda_1(\mathbf{1}^{\top} \bm{\beta}-1)+\lambda_2(\mu_0- \bm{\beta}^{\top}\bm{\mu})-\bm{\gamma}^{\top} \bm{\beta}.
\end{align}
We have
\begin{equation*}
  \mathbf{1}^{\top} \bm{\beta}=1,\ \bm{\mu}^{\top} \bm{\beta}\ge \mu_0,\  \bm{\beta} \ge 0,
\end{equation*}
\begin{equation*}
  \lambda_2\ge 0,\ \bm{\gamma}\ge 0,
\end{equation*}
\begin{equation*}
\lambda_2(\mu_0-\bm{\mu}^{\top} \bm{\beta})=0,\ \gamma_i\beta_i = 0,i = 1,2,\cdots,n.
\end{equation*}
And KKT first-order necessary conditions is 
\begin{align*}
  \nabla_{\beta}\mathcal{L}( \bm{\beta}) = 0,
\end{align*}
Define that $A = \left\{i :\beta_i>0 \right\}$, $I = \left\{i : \beta_i=0\right\}$. For $i \in A$, we have $\gamma_i =0$. For $i \in I$, we have $\gamma_i \ge 0$.
We can obtain 
\begin{align}
   \bm{\beta}_A &= \Sigma^{-1}_{AA}\left(\lambda_1\mathbf{1}_A + \lambda_2\bm{\mu}_A\right),\label{eq:beta_A}\\
   \bm{\beta}_I &= \mathbf{0} \label{eq:beta_I}.
\end{align}

\begin{enumerate}
  \item When $i \in A$,
  $\Sigma_{AA} \bm{\beta}_{A} = \lambda_1\mathbf{1}_{A} +\lambda_2\bm{\mu}_{A}$.
    \begin{itemize}
      \item When $\lambda_2 = 0$.
      \begin{align*}
        \Sigma_{AA} \bm{\beta}_{A} &= \lambda_1\mathbf{1}_A \\
        \mathbf{1}_A^{\top} \bm{\beta}_{A} &= 1,
      \end{align*}
     we have 
     \begin{equation}\label{eq:bata when lambda2=0}
        \bm{\beta}_A = \frac{\Sigma_{AA}^{-1}\mathbf{1}_A}{\mathbf{1}_A^{\top}\Sigma_{AA}^{-1}\mathbf{1}_A}.
     \end{equation}
      \item When $\lambda_2 \neq 0$.
      \begin{align*}
            \left\{\begin{array}{l}
    \Sigma_{AA} \bm{\beta}_A = \lambda _1 \mathbf{1}_{A} + \lambda _2\bm{\mu}_{A},\\
     \bm{\mu}_{A}^{\top} \bm{\beta}_A = \mu_0,\\
     \mathbf{1}_{A}^{\top} \bm{\beta}_{A} = 1.
    \end{array}\right.
      \end{align*}
      We have 
      \begin{align}
            \left\{\begin{array}{l}
     \lambda_2 = \frac{\bm{\mu}_A^{\top}\Sigma_{AA}^{-1}\mathbf{1}_A-\mu_0\mathbf{1}_{A}^{\top}\Sigma_{AA}^{-1}\mathbf{1}_A}{\left(\bm{\mu}_A^{\top}\Sigma_{AA}^{-1}\mathbf{1}_A\right)^2+\bm{\mu}_{A}^{\top}\Sigma
     _{AA}^{-1}\mu_A\mathbf{1}_A^{\top}\Sigma_{AA}^{-1}\mathbf{1}_A},\\
         \lambda_1 = \frac{1-\lambda_2\mathbf{1}_A^{\top}\Sigma_{AA}^{-1}\bm{\mu}_A}{\mathbf{1}_A^{\top}\Sigma_{AA}^{-1}\mathbf{1}_A},\\
         \bm{\beta}_A = \lambda_1\Sigma_{AA}^{-1}\mathbf{1}_{A}+\lambda_2\Sigma_{AA}^{-1}\bm{\mu}_A.
    \end{array}\right.
      \end{align}
    \end{itemize}
  \item When $i \in {I}$, $\beta_{i} = 0$.
      \begin{equation*}
        \left(\Sigma\beta\right)_i = \lambda_1 + \lambda_2\mu_i - \gamma_i,
      \end{equation*}
      we have 
      \begin{equation}\label{eq:gamma}
        \gamma = \left\{\begin{array}{l}
0,\ i \in A,
 \\
\lambda _1\mathbf{1}_{I} + \lambda _2\mu_{I} - \Sigma_{I_{A}}\bm{\beta}_A, \ i \in I. 
\end{array}\right.
      \end{equation}
\end{enumerate}
If $\gamma_i \ge 0$, then the active set assumption holds.

The following algorithm (cf. Alg. \ref{alg:The determination of the active set}) describes the procedure for determining the active set.
\begin{algorithm}[htbp]
    \caption{The determination of the active set}
    \label{alg:The determination of the active set}
    \begin{algorithmic}[1] 
        \STATE \textbf{We assume no non-negativity constraints, solve that $\beta^*(w) = \Sigma^{-1}\left(\lambda_1\mathbf{1} + \lambda_2\mu\right)$.}
        
        \IF {\textbf{$\beta_i \ge 0$,$\forall i$},}
            \STATE Active set $A = \{1,2,\cdots,n \}$.
        
        \ELSIF {$\exists k$ such that $\beta_k^* < 0$,}
            \STATE Let $\beta_k=0$, $k\in I$.
        \ENDIF
        \STATE Solve Eq. \eqref{eq:beta_A} and Eq. \eqref{eq:gamma}.
        \IF {$\exists j \in A$, $\beta_j <0$,}
        \STATE Move $j$ from $A$ into $I$.
        \ENDIF
        \IF {$\exists i \in I$, $\gamma_i <0$,}
        \STATE Move $i$ form $I$ into $A$.
        \ENDIF
        \STATE Repeat the above steps until $\beta_A >0 $ and $\gamma_I \ge 0$, where $A$ is then correct.
    \end{algorithmic}
\end{algorithm}

By virtue of the aforementioned method, we have derived the analytical expression for the optimal solution of Eq. \eqref{eq:MO-MV eq}. Since both the expected return and the covariance matrix of the portfolio are functions of the risk factor $w$, the optimal solution expression is essentially a polynomial function with respect to $w$. Consequently, given the current expected returns of individual stocks and the upper-lower covariance matrices, there exists a unique optimal risk factor $w$ that minimizes the model's objective function.

\section{Simulation}\label{sec:Simulation}
In this section, we conduct empirical analysis using three sets of experimental data: one set of simulated generated data, and the other two sets of representative US stock data and A-share data respectively. These data are employed to test the practical performance of the proposed model (SLE-MUV model), which is further compared with the traditional mean-variance model (MV model). The empirical results demonstrate that investors can flexibly select the risk factor \( w \) in the SLE-MUV model according to the current market environment and their own investment preferences, thereby determining the optimal investment proportions.
\subsection{How to obtain mean, covariance matrix with uncertainty}\label{sec:mean_covariance}
 Now, we adopt the moving block method proposed by \cite{Yang2023} to estimate the minimum and maximum variances of data following a G-normal distribution.
 Assuming that a asset return $X^{(k)}\sim \mathcal{N}(\mu^{(k)}, [\underline{\sigma}^{(k)^2},\overline{\sigma}^{(k)^2}]),\ k=1,2,\dots N $ in time interval $T_0$ is $\left(x^{(k)}_i\right)_{i=1}^{T}$, the number of return is T. Give a block length $n_1$,  $\left(x^{(k)}_i\right)_{i=1}^{T}$ are scanned subsequently as $m = T-n_1+1$ blocks,
\begin{align}\label{eq: m blocks}
  \{1,2, \dots, n_1\},\ \{2,3, \dots, n_1+1\}, \dots, \{T-n_1+1, T-n_1+2, \dots, T\}.
\end{align}
Denote the data in the $l$th block by $B^{(k)}_l = \{(x^{(k)}_i)\}_{l\le i \le l+n_1-1,\ 1\le l\le m}$, assuming that $B^{(k)}_{l} \sim \mathcal{N}(\mu_l,\sigma_l^{(k)^2}) = \mu^{(k)}_l + \epsilon^{(k)}_l$, and $\epsilon^{(k)}_l \sim \mathcal{N} (0,\sigma_l^{(k)^2})$.

\begin{enumerate}
  \item \textbf{Estimators for the parameters $\mu^{(k)}, \underline{\mu}^{(k)}, \overline{\mu}^{(k)}$}.
\begin{align}\label{eq: mean of blocks}
  \tilde{\mu}^{(k)} = \frac{1}{T}\sum_{i=1}^{T}x_{i}^{(k)}.
\end{align}
The lower and upper means $\{\underline{\mu}^{(k)}, \overline{\mu}^{(k)}\}$ are estimated, respectively, by
\begin{align}\label{eq: lower and upper mean}
 \hat{ \underline{\mu}} ^{(k)}= \min_{1\le l \le m} \tilde{\mu}^{(k)}_{l},\quad \hat{\overline{\mu}}^{(k)} = \max_{1\le l \le m} \tilde{\mu}^{(k)}_{l}.
\end{align}

\item \textbf{Estimators for the parameters $\underline{\sigma}^{(k)^2}, \overline{\sigma}^{(k)^2}$}.
    
   \begin{align}\label{eq: variance of epsilon}
  \hat{\sigma}_l^{(k)^2} = \frac{1}{n_1-1}\sum_{i=l}^{l+n_1-1}{(x_i^{(k)}-\mu^{(k)}_l)}^2,
\end{align}
and then the lower variance $\underline{\sigma}^{(k)^2}$ is estimated by
\begin{align}\label{eq: lower variance}
  \underline{\sigma}^{(k)^2} = \min_{1\le l\le m}(\hat{\sigma}^{(k)}_l)^2.
\end{align}
Let $n_2\le n_1$, 
\begin{align}\label{eq: pre-calculate upper variance}
  \tilde{x}^{(k)}_i = x^{(k)}_i - \frac{1}{n_2}\sum_{i = 1+(j-1)n_2}^{jn_2}x^{(k)}_i, \ 1+(j-1)n_2 \le i \le jn_2, \ 1 \le j \le T/n_2.
\end{align}
For $1\le l \le m$,
\begin{align}\label{eq: variance of miniblocks}
\hat{\tilde{\sigma}}_l^{(k)^2} = \frac{1}{n_1-1}\sum_{i=l}^{l+n_1-1}(\tilde{x}^{(k)}_i)^2
\end{align}
we estimate the upper variance by
\begin{align}\label{eq:upper variance}
 \hat{\overline{\sigma}}^{(k)^2} = \max_{1\le j\le m}(\hat{\tilde{\sigma}}^{(k)})^2.
\end{align}

\item  \textbf{Estimators for the parameters $\overline{\sigma}^{(ij)^2}$, $ \underline{\sigma}^{(ij)^2}$}.
\begin{align}\label{eq: formulate of pho}
  \overline{\sigma}^{(ij)^2} &= \mathbb{E}[X^{(i)}X^{(j)}]-\mu^{(i)}\mu^{(j)}\nonumber\\
  \underline{\sigma}^{(ij)^2} &= -\mathbb{E}[-X^{(i)}X^{(j)}]-\mu^{(i)}\mu^{(j)}.
\end{align}
For $X^{(i)}X^{(j)}$, we have $X^{(ij)} = \left(x^{(i)}_k x^{(j)}_k\right)_{k=1}^{T}$, we can estimate the parameters $\underline{\sigma}^{(ij)^2} = \underline{\mu}^{(ij)}$, $\overline{\sigma}^{(ij)^2} = \overline{\mu}^{(ij)}$, where $\underline{\mu}^{(ij)}$, $\overline{\mu}^{(ij)}$ are the lower and upper means of $X^{(ij)}$ respectively.
\end{enumerate}
Based on the above estimation method, we can derive the portfolio mean, as well as the element-wise maximum and minimum matrices of the covariance matrix.
The logic of the solving algorithm in this article is as following:

First, stocks are selected in accordance with the preset target screening criteria to construct an initial investment portfolio. Second, the mean return of each stock in the portfolio and the return covariance matrix of the portfolio are calculated, followed by a positive definiteness test on the covariance matrix. 
If the covariance matrix is non-positive definite, stocks are reselected to update the portfolio, and the above steps of calculating the mean return and covariance matrix are repeated. 
If the covariance matrix satisfies the positive definiteness requirement, the cvxpy algorithm\citep{diamond2016cvxpy} is adopted to solve the problem of optimal investment weight allocation.

In the solution stage of the cvxpy algorithm, if a valid feasible solution is obtained, the optimal investment weights are directly output. If no feasible solution is found, stocks are reselected to update the portfolio, and the process backtracks to the step of covariance matrix calculation. 
If repeated attempts at stock reselection all fail, the current investment portfolio is retained, and a positive-definiteness fine-tuning correction is performed on its covariance matrix; after correction, the cvxpy algorithm is applied again for solution.

If no valid feasible solution can be derived from the cvxpy algorithm under all scenarios, the solution algorithm is switched to the Sequential Least Squares Quadratic Programming (SLSQP) algorithm\citep{2020SciPy-NMeth} to obtain a local optimal solution for the portfolio optimization problem, which is then output to complete the entire algorithm process. The specific algorithmic process is detailed in Alg. \ref{alg:portfolio_optimization}.

\begin{algorithm}[htbp]
    \caption{Investment Portfolio Optimal Weight Solving Algorithm}
    \label{alg:portfolio_optimization}
    \begin{algorithmic}[1] 
        \STATE \textbf{Step 1: Construct Initial Portfolio.}
        \STATE Select stocks according to preset target screening criteria to build the initial investment portfolio.
        
        \STATE \textbf{Step 2: Calculate Portfolio Statistical Indicators.}
        \STATE Compute the mean return of each stock in the portfolio and the return covariance matrix $\Sigma$ of the portfolio.
        
        \STATE \textbf{Step 3: Positive Definiteness Test of Covariance Matrix.}
        \STATE Perform a positive definiteness test on $\Sigma$:
        \IF {$\Sigma$ is non-positive definite,}
            \STATE Reselect stocks to update the portfolio, and \textbf{go back to Step 2}.
        \ENDIF
        
        \STATE \textbf{Step 4: Solve with cvxpy Algorithm.}
        \STATE Execute the cvxpy algorithm for solution:
        \IF {A valid feasible solution is obtained,}
            \STATE Output the optimal investment weights directly, and terminate the algorithm.
        \ELSIF {No feasible solution is found}
            \STATE Reselect stocks to update the portfolio, and \textbf{go back to Step 2}.
        \ENDIF
        
        \STATE \textbf{Step 5: Covariance Matrix Positive-Definiteness Fine-Tuning.}
        \IF {repeated attempts at stock reselection all fail,}
            \STATE Retain the current investment portfolio, perform positive-definiteness fine-tuning correction on $\Sigma$ to make it positive definite, and \textbf{go back to Step 4}.
        \ENDIF
        \STATE \textbf{Step 6: Switch to SLSQP Algorithm for Local Optimal Solution.}
        \IF { no valid feasible solution can be derived from cvxpy under all scenarios,}
            \STATE Switch the solution algorithm to SLSQP, solve for the locak optimal solution of the portfolio optimization problem, output the result, and terminate the algorithm.
        \ENDIF
    \end{algorithmic}
\end{algorithm}

\subsection{Empirical Study on Synthetic Data}\label{sec: Synthetic Data}
To further validate our conclusion, we now use the method by \cite{Yang2023} to construct $N$ asset returns $\left(X^{(i)}_j\right)_{j=1}^{T}$,$i = 1,2,...,N$ that follow different G-normal distribution $\mathcal{N}\left(\mu_i,[\underline{\sigma}^{(i)^2},\overline{\sigma}^{(i)^2}]\right)$, and use Algorithm \ref{alg:portfolio_optimization} to solve the optimal solution of the Mean-Variance model under sublinear expectations. 

\textbf{Data generation process}.
       The sample $\left(X^{(i)}_j\right)_{j=1}^{T}$ satisfy
       \begin{equation*}
         X_j^{(i)} = \epsilon_{j},\ 1+ n_0(k-1) \le j \le n_0 k,\  1\le k \le K.
       \end{equation*}
       where $\epsilon_{j} \sim \mathcal{N}(\mu_i,\sigma_{j}^{(i)^2})$,$1+ n_0(k-1) \le j \le n_0 k$ with $\sigma_{j}^{(i)^2} \in [\underline{\sigma}^{(i)^2},\overline{\sigma}^{(i)^2}]$. Specifically,
        $T$ pieces of data are divided into $K$ groups, each group containing $n_0$ pieces of data, following a normal distribution $\mathcal{N}(\mu_i,\sigma_{j}^{(i)^2})$.
        
Assuming $N=4$ and with reference to the daily return data of real stocks, we set the mean of the synthetic data to (0.0011, 0.00033, 0.00062, 0.00056), the upper variance to (0.000676, 0.000196, 0.0004, 0.000625) and the lower variance to (0.000484, 0.0001, 0.000289, 0.0004), respectively. After determining the expected returns for each asset, let $K = 4$, $n_0 = 250$, we obtain four sample sets that follow different G-normal distributions and the upper and lower variances $\hat{\overline{\sigma}}^{(i)^2}$, $\hat{\underline{\sigma}}^{(i)^2}$ of each G-normal distribution by Subsection \ref{sec:mean_covariance}. In this experiment, the mean returns and upper and lower variances of the four assets we obtained are shown in Table \ref{tab:mean Variance}.
\begin{table}[htbp]
    \centering
    \caption{Mean and upper and lower variance of returns on four assets (values multiplied by $10^{-4}$)}
    \label{tab:mean Variance}
    \sisetup{
        round-mode=places,
        round-precision=4,
        table-format=-1.4
    }
    \begin{tabular}{c *{4}{S}}
        \toprule
        Asset & {$\mu$ ($\times 10^{-4}$)} & {$\hat{\sigma}^2$ ($\times 10^{-4}$)} & {$\hat{\overline{\sigma}}^2$ ($\times 10^{-4}$)} & {$\hat{\underline{\sigma}}^2$ ($\times 10^{-4}$)} \\
        \midrule
        $X^{(1)}$ & 21.7485 & 5.3408 & 5.6973 & 4.4581  \\ 
        $X^{(2)}$ &  1.4013 & 1.6485 & 2.2457 & 1.1721  \\ 
        $X^{(3)}$ & 3.8578 & 3.2588 & 4.3433 & 2.2867  \\ 
        $X^{(4)}$ & -20.2246 & 6.0732 & 6.7606 & 5.2222  \\ \hline
    \end{tabular}
\end{table}
The upper and lower covariance matrices of synthetic data are
\begin{align*}
\small{
  \overline{V} = 10^{-4} \cdot \begin{pmatrix}
 5.6973  & 0.3749  & 0.4031  & 0.5626   \\ 
 0.3749  & 2.2457  & 1.7837  & 0.5936   \\ 
 0.4031  & 1.7837  & 4.3433  & 0.7980   \\ 
 0.5626  & 0.5936  & 0.7980  & 6.7606   \\ 
\end{pmatrix}},\ 
\small{\underline{V} = 10^{-4} \cdot \begin{pmatrix}
 4.4581  & -0.6167  & -0.8375  & -0.4505   \\ 
 -0.6167  & 1.1721  & -1.1163  & -0.6402   \\ 
 -0.8375  & -1.1163  & 2.2867  & -0.6523   \\ 
 -0.4505  & -0.6402  & -0.6523  & 5.2222   \\ 
\end{pmatrix}}.
\end{align*}

By Algorithm \ref{alg:portfolio_optimization}, a weighted convex optimization problem \eqref{eq:SP-MV eq} equivalent to the multi-objective mean variance model \eqref{eq:MO-MV eq} is solved, and the optimal solution is obtained when the risk factor is set to [0,1], thereby obtaining the Pareto frontier curve of the multi-objective optimization problem \eqref{eq:MO-MV eq} on the first sample day, as shown in the Figure \ref{fig: Pareto_frontier}.
\begin{figure}[htbp]
  \centering
  \includegraphics[width=0.6\textwidth]{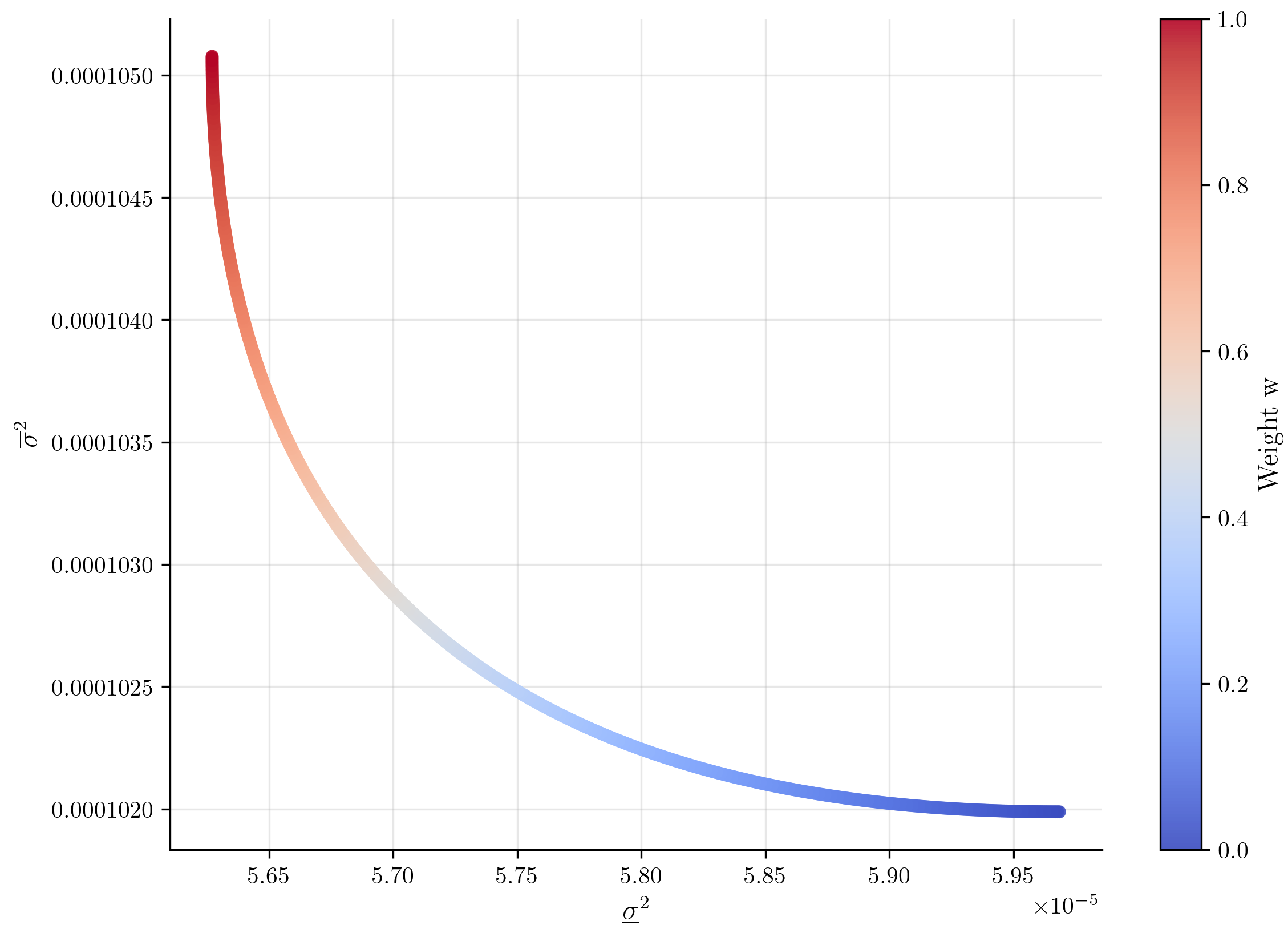}
  \caption{Pareto frontier of portfolios constructed with synthetic data on the first sample day}\label{fig: Pareto_frontier}
\end{figure}

We also compared it with the classical mean variance model, where the covariance matrix is
\begin{equation*}
  \hat{V} = 10^{-4} \cdot \begin{pmatrix}
  5.6138  & -0.0533  & -0.0064  & -0.0204   \\ 
        -0.0533  & 1.6001  & -0.0137  & 0.0268   \\ 
        -0.0064  & -0.0137  & 3.1433  & -0.1179   \\ 
        -0.0204  & 0.0268  & -0.1179  & 5.5147   \\ 
\end{pmatrix}.
\end{equation*}
 
For each portfolio selection model, we determine the initial optimal portfolio weight vector $\bm{\beta}^*$ using the first 252 trading days as in-sample data. The first out-of-sample portfolio return is calculated as $\bm{\beta}^{*^{\top}}\bm{\mu}$, where $\bm{\mu}$ denotes the daily return vector of selected stocks on the 253th trading day (the first out-of-sample trading day). 
A rolling window approach is then adopted to slide the window forward, with out-of-sample returns computed for all consecutive trading days over the following 500 days.

The optimal configuration at 500th sample day can be obtained as 0.1803, 0.4355,0.2702,0.1140, with a return of 0.00365 and a variance of 0.0000922 at End Sample. 
By comparing SLE-MUV model with MV model, we found that when $w=0.523$, the calcumulaive wealth value of the two models' investment portfolios under optimal allocation at 500th sample day are almost equal, indicating the adaptability of the SLE-MUV model proposed in this paper, which can replace the MV model to allocate investment portfolios. 
Another advantage of this model is that investors can determine the weight $w$ based on their own investment habits and the current investment market environment, and directly obtain the optimal allocation of assets. 
\begin{figure}[H]
  \centering
  \begin{subfigure}[htbp]{0.49\textwidth}
  \centering
  \includegraphics[width=1\textwidth]{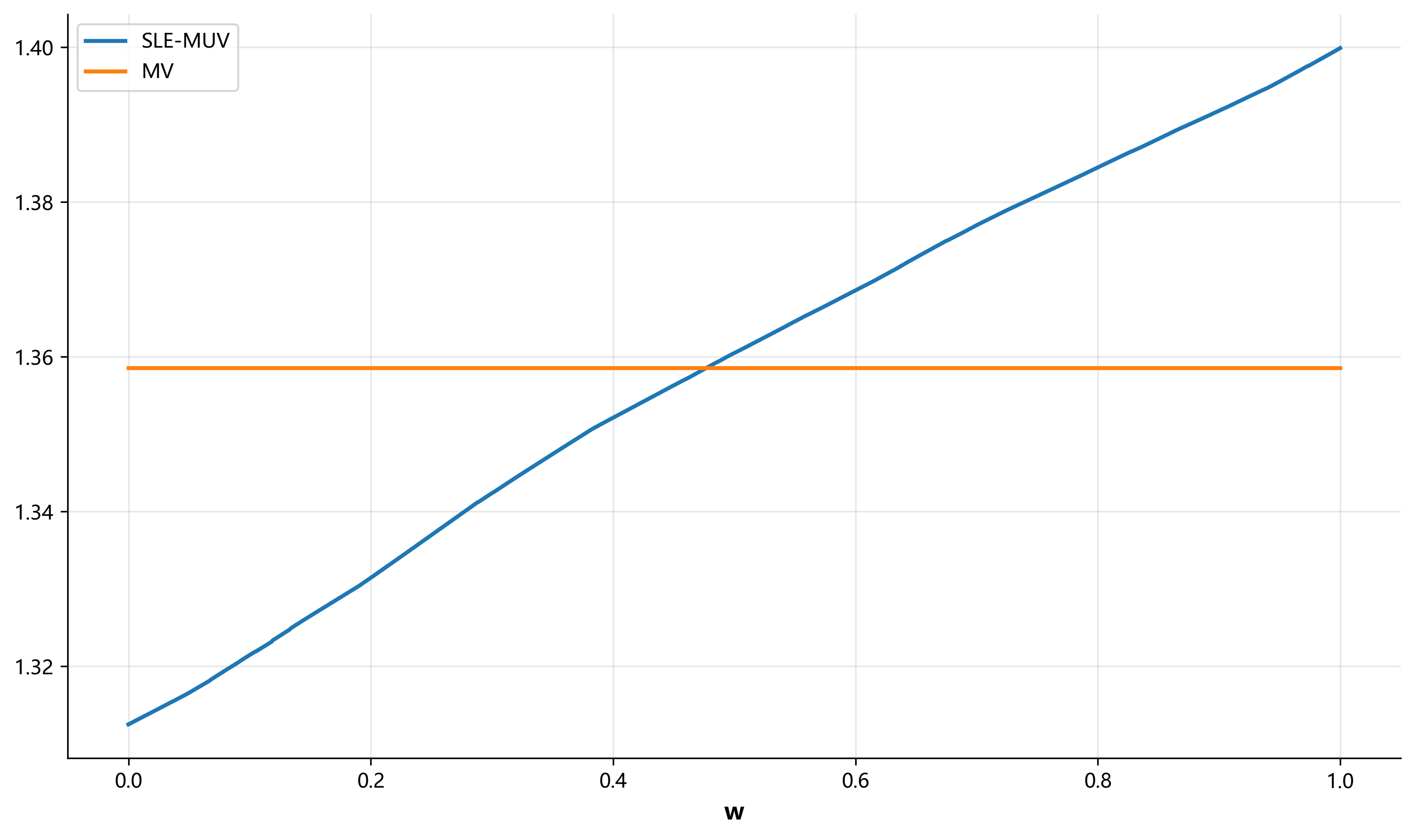}
  \caption{Cumulative Wealth at End Sample}\label{Cumulative Wealth at End Sample}
\end{subfigure}
\begin{subfigure}[htbp]{0.49\textwidth}
  \centering
  \includegraphics[width=1\textwidth]{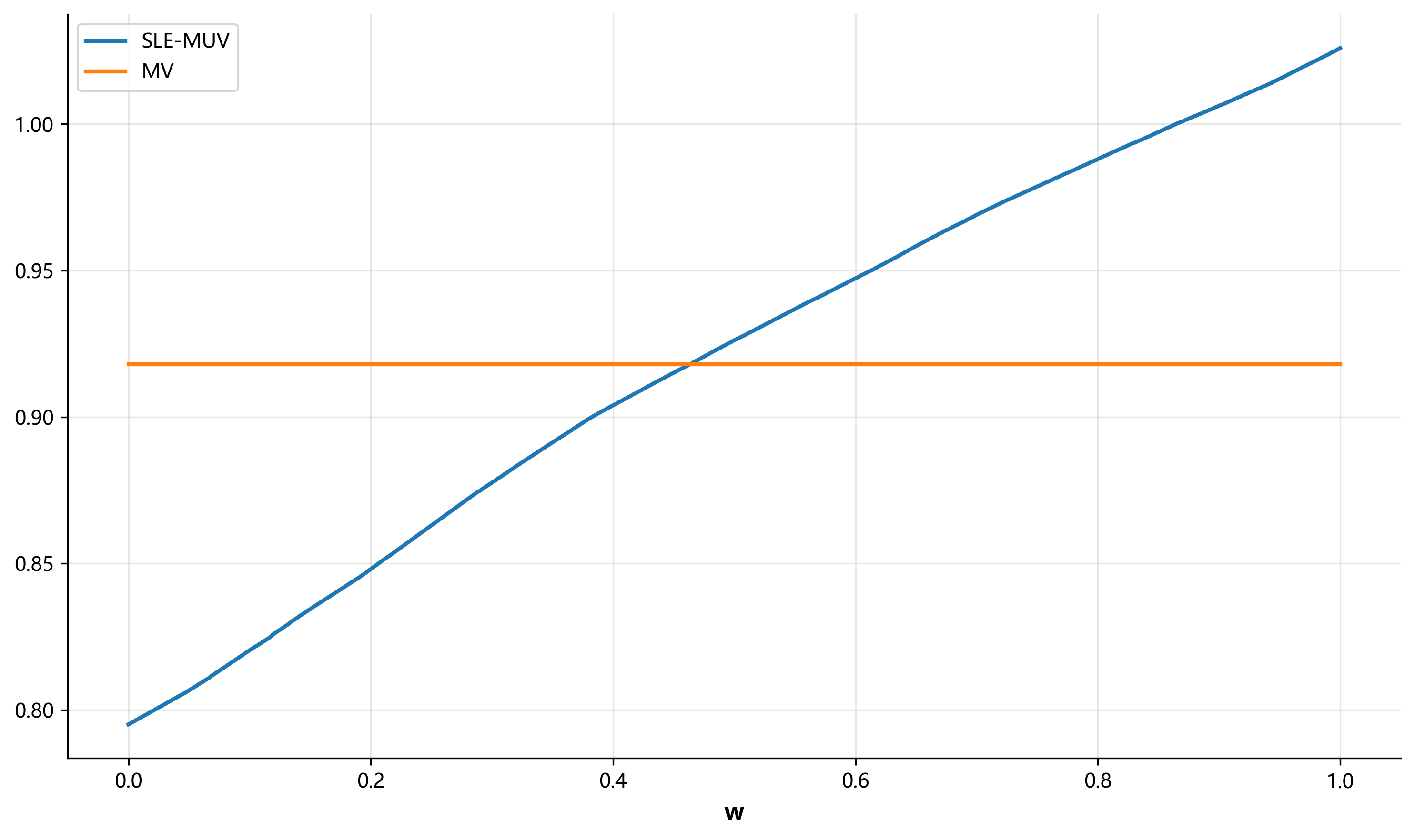}
  \caption{Sharpe Ratio}\label{Sharpe Ratio-Synthetic}
  \end{subfigure}
\caption{Cumulative Wealth as of End Sample,Sharpe Ratio: SLE-MUV Model vs. Traditional MV Model}
  \end{figure}

Figure \ref{Cumulative Wealth at End Sample} illustrates the evolution of cumulative wealth for both SLE-MUV model and MV model as a function of parameter $w$ ($0\le w \le 1$) at the end of the sample period. 
The traditional MV model serves as the benchmark, maintaining a stable cumulative wealth level of 1.3585 throughout the parameter range. 
The SLE-MUV model exhibits a sustained upward trend in cumulative wealth as $w$ increases at end sample.
Specifically, for $w \in [0,0.522]$, the cumulative wealth of the SLE-MUV model is slightly lower than that of the traditional MV model.
For $w \in [0.523,1]$, the SLE-MUV model outperforms the MV model and the performance advantage widens continuously. At $w = 1$, the cumulative wealth of the SLE-MUV model peaks at 1.3998, representing a 0.0413 increase relative to the traditional MV model, highlighting its superior ability to capture excess returns in the high-parameter region.

As shown in Figure \ref{Sharpe Ratio-Synthetic}, the Sharpe Ratio of the SLE-MUV model is closely correlated with its cumulative wealth trend, and it gradually increases with the rise of the risk factor w. The Sharpe Ratio of the MV model is 0.918, while the Sharpe Ratio of the SLE-MUV model increases from the initial value of 0.7951 to a peak of 1.0259 when w=1, which is significantly higher than that of the benchmark model. From a portfolio management perspective, the higher Sharpe ratio of the SLE-MUV model implies that, at equivalent risk levels, it can generate greater excess returns by optimizing portfolio weights, demonstrating superior risk pricing efficiency. This enhances the quality of portfolio returns without amplifying systematic risk.

The traditional MV model achieves excellent downside risk control with a maximum drawdown of 10.33\%. In contrast, the SLE-MUV model consistently exhibits a larger absolute maximum drawdown, which decreases from 12.33\% to 11.17\% as $w$ increases. However, the difference between the two models remains moderate, and there is no evidence of extreme risks.
This pattern is consistent with the basic financial principle of ``higher return corresponding to higher risk'': the SLE-MUV model assumes more prominent tail risk to achieve greater cumulative wealth and Sharpe Ratio. Although its downside risk exceeds that of the traditional MV model, the marginal increase in risk premium is more significant.

Investors can dynamically adjust the factor risk $w$ to optimize the risk-return tradeoff: increasing $w$ during market upswings to amplify returns, and reducing $w$ during market downturns to mitigate risk, thereby achieving dynamic portfolio optimization.
\begin{figure}[H]
  \centering
  \begin{subfigure}[b]{0.49\textwidth}
  \includegraphics[width=1\textwidth]{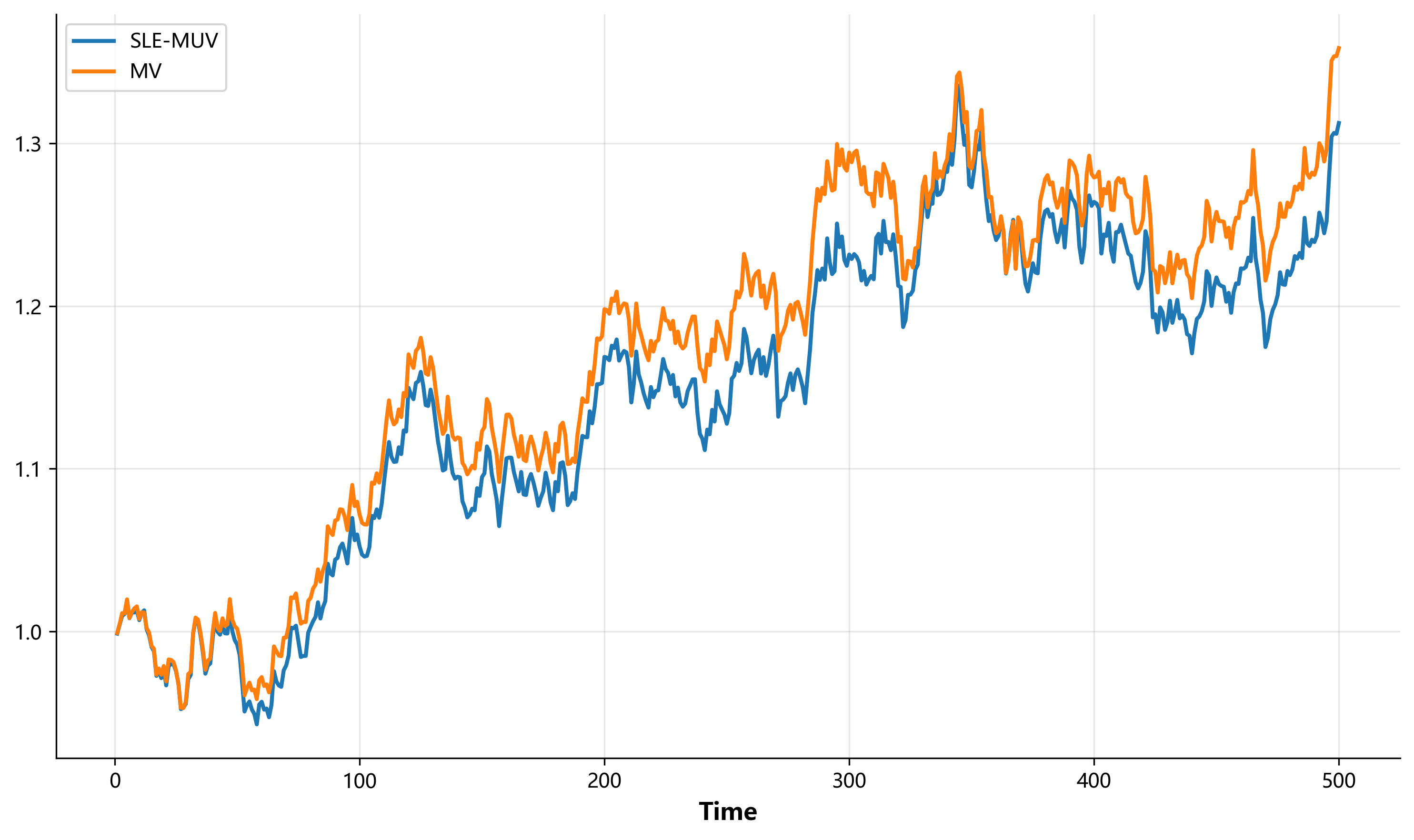}
  \caption{w=0.}\label{Synthetic Data w=0}
\end{subfigure}
\begin{subfigure}[b]{0.49\textwidth}
  \centering
  \includegraphics[width=1\textwidth]{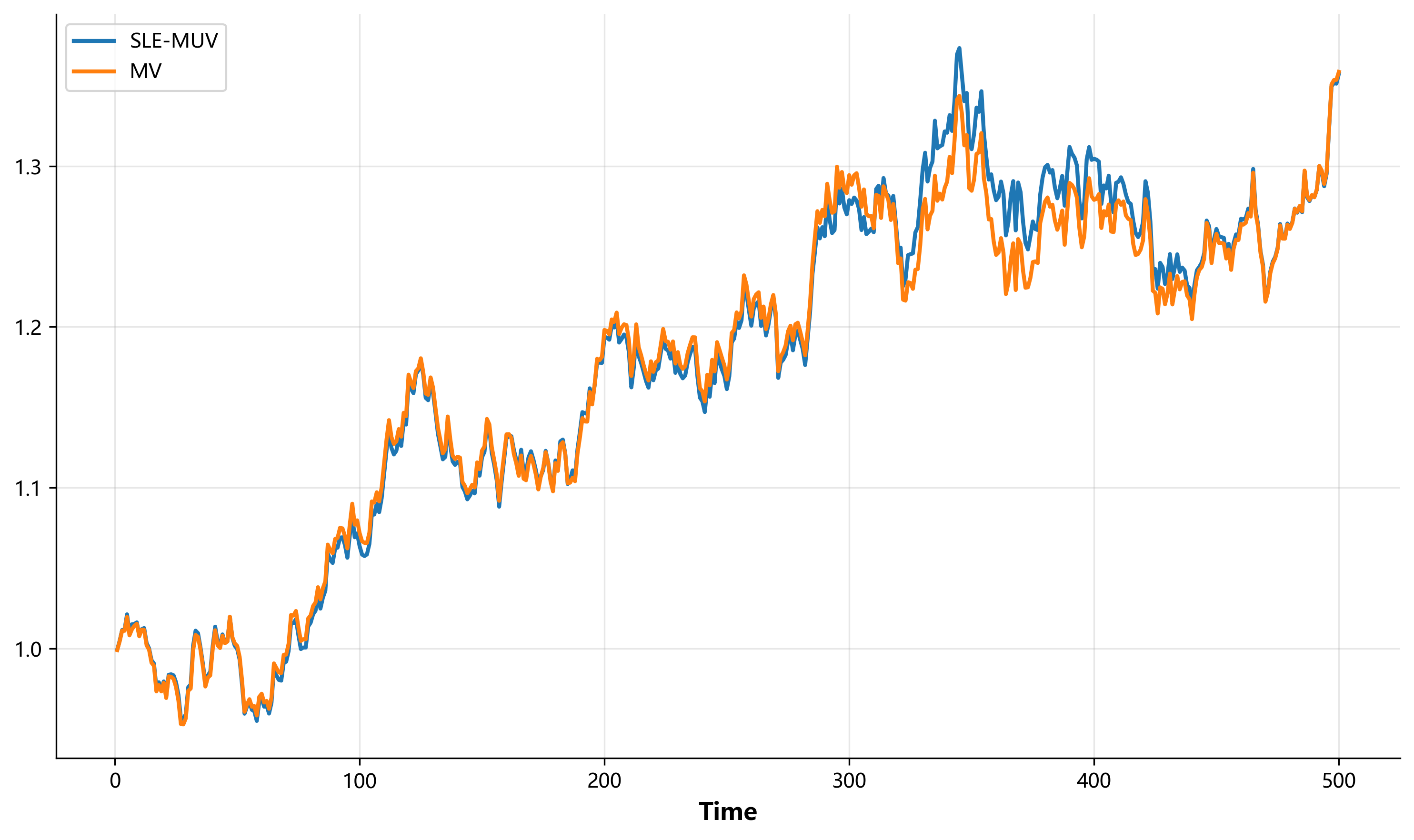}
  \caption{w=0.464.}\label{Synthetic Data w=0.464}
\end{subfigure}
\begin{subfigure}[b]{0.49\textwidth}
  \centering
  \includegraphics[width=1\textwidth]{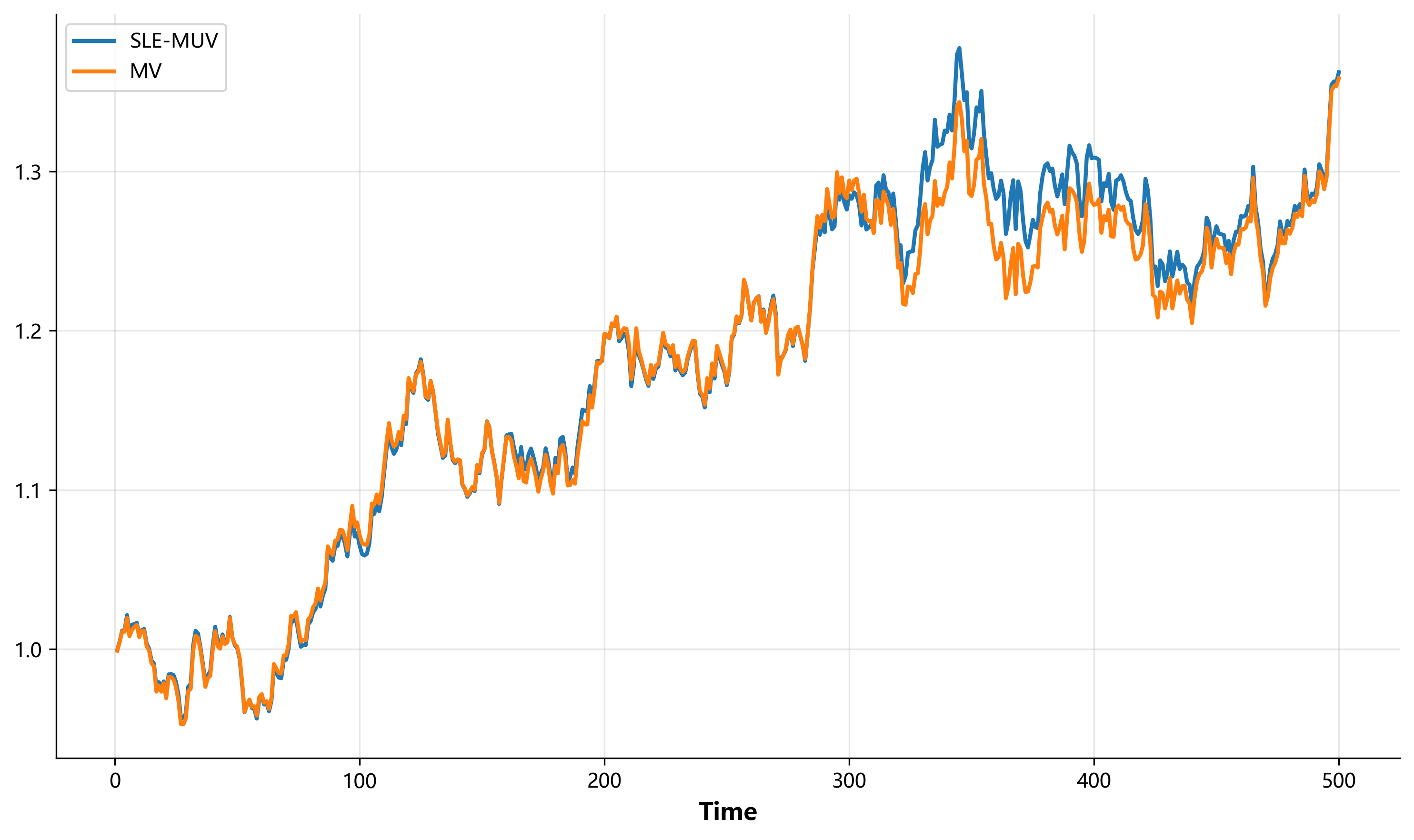}
  \caption{w=0.523.}\label{Synthetic Data w=0.523}
  \end{subfigure}
  \begin{subfigure}[b]{0.49\textwidth}
  \centering
  \includegraphics[width=1\textwidth]{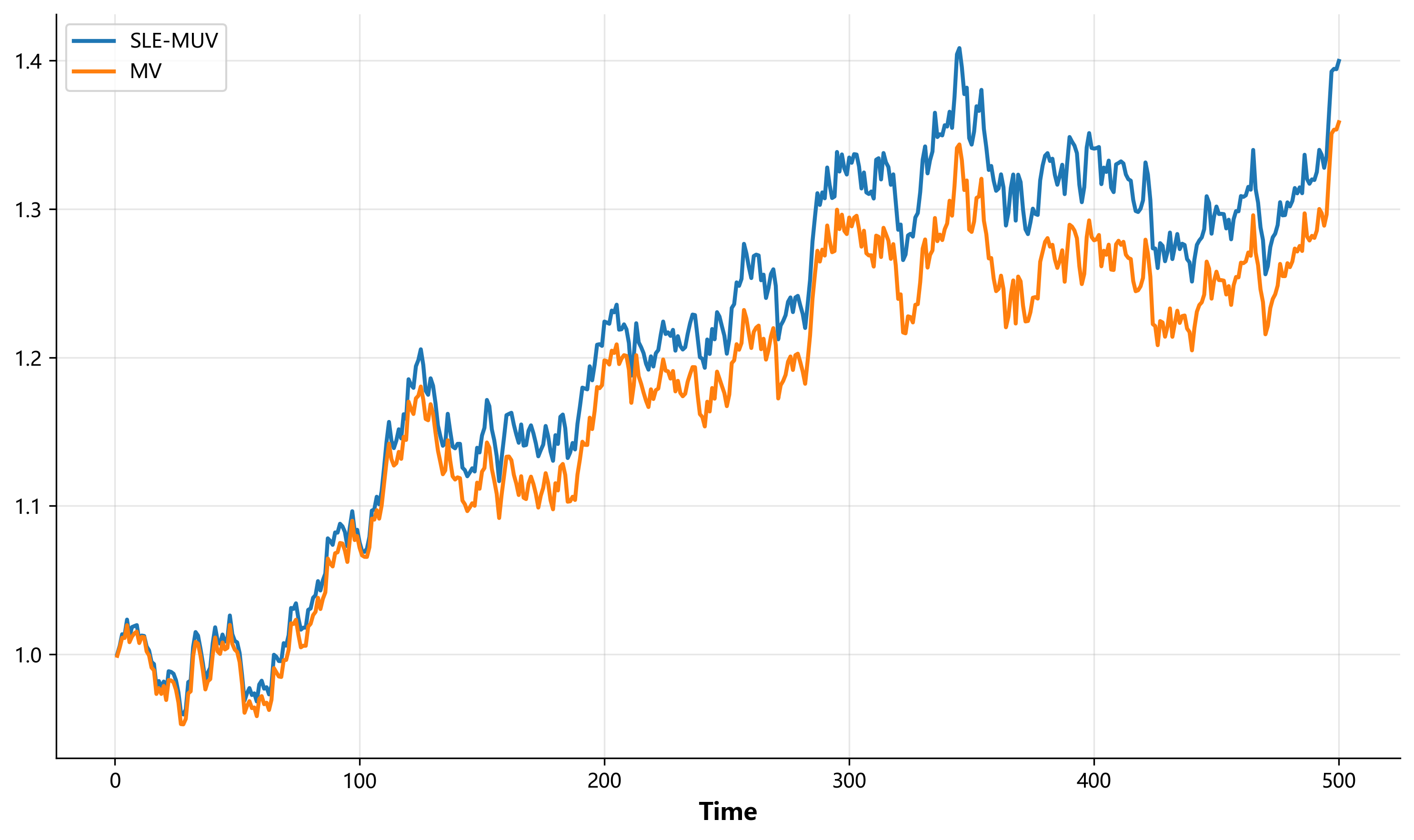}
  \caption{w=1.}\label{Synthetic Data w=1.0}
  \end{subfigure}
  \caption{Protfolio Cumulative Wealth for SLE-MUV model and traditional MV Model under Different w Values.\protect\footnotemark}
  \label{fig:Sythetic data wealth fig}
\end{figure}
  \footnotetext{MV-SLE and MV-LE denote the curves of the SLE-MV model and the conventional MV model, respectively.}
Figure \ref{fig:Sythetic data wealth fig} presents the cumulative wealth curves of the investment portfolio. We select the parameter values of $w$ that correspond to the minimum and maximum differences in maximum drawdown, as well as the points where the cumulative wealth and Sharpe Ratio  of the SLE-MUV model equals from that of the traditional MV model, namely $w = 0, 1, 0.464, 0.523$, for comparative analysis.
As illustrated in Figure \ref{fig:Sythetic data wealth fig}, over the 500-trading-day sample period, the cumulative wealth curves of both models exhibit a highly similar overall fluctuation pattern. However, as the risk factor $w$ gradually decreases, the advantage of the SLE-MUV model in cumulative wealth becomes increasingly pronounced, which is fully consistent with the results presented in Figure \ref{Cumulative Wealth at End Sample}. For investors with strong risk-bearing capacity, a long investment horizon, and no pursuit of short-term market fluctuations, they can select the risk factor $w \in [0.523,1]$, especially $w=1$, to construct an investment portfolio that meets their own needs and achieve the expected investment objectives.

The experimental results indicate that the optimal value of the risk factor $w$ is sensitive to generated samples and varies with changes in generated data. Accordingly, its optimal parameter should be selected based on the out-of-sample performance of the portfolio. In this paper, we construct investment portfolios and conduct empirical analyses using representative individual stocks from the U.S. and A-share markets.  The empirical results demonstrate that the evolutionary relationships between the three core evaluation indicators (cumulative wealth, Sharpe ratio and maximum drawdown) and the risk factor w exhibit obvious market heterogeneity. On this basis, the optimal value of the risk factor can be determined according to the actual portfolio performance under different $w$ settings, which further guides the asset weight allocation of the portfolio.

\subsection{Empirical Study on Real Data}
In this subsection, to verify the effectiveness and superiority of the SLE-MUV model, we select 6 representative U.S. stocks and 6 A-shares as experimental targets to conduct backtesting simulation experiments. The experimental results fully demonstrate that, compared with the traditional MV model, by flexibly adjusting and selecting the risk factor w (with a value range of [0,1]) of the SLE-MUV model, investors can not only significantly improve the cumulative wealth level and Sharpe ratio of the investment portfolio, realize the optimization of return capacity, but also stably control the gap in maximum drawdown between the SLE-MUV model and the MV model within a reasonable and controllable range, balance returns and risks, and highlight the application value and advantages of the SLE-MUV model in practical investment decisions.
\subsubsection{US stocks}\label{sec:US stocks}
 We selected six representative U.S. stocks --- Apple (AAPL), NVIDIA (NVDA), Microsoft (MSFT), Johnson \& Johnson (JNJ), Pfizer (PFE), and Merck \& Co. (MRK) --- to construct our investment portfolio and validate the feasibility and superiority of the proposed method based on real stock data. 
The empirical time horizon covers the period from January 2, 2019, to December 25, 2025. 
The daily return is calculated in accordance with the definition proposed in \cite{cai2025}: for a stock on the $t$-th trading day, its daily return $\mu_t$ is defined as $\mu_t = (P_{t+1}-P_{t})/P_t$, Where $P_t$ denotes the adjusted closing price of the stock on the $t$-th trading day.
\begin{figure}[htbp]
  \centering
  \begin{subfigure}[htbp]{0.49\textwidth}
  \centering
  \includegraphics[width=1\textwidth]{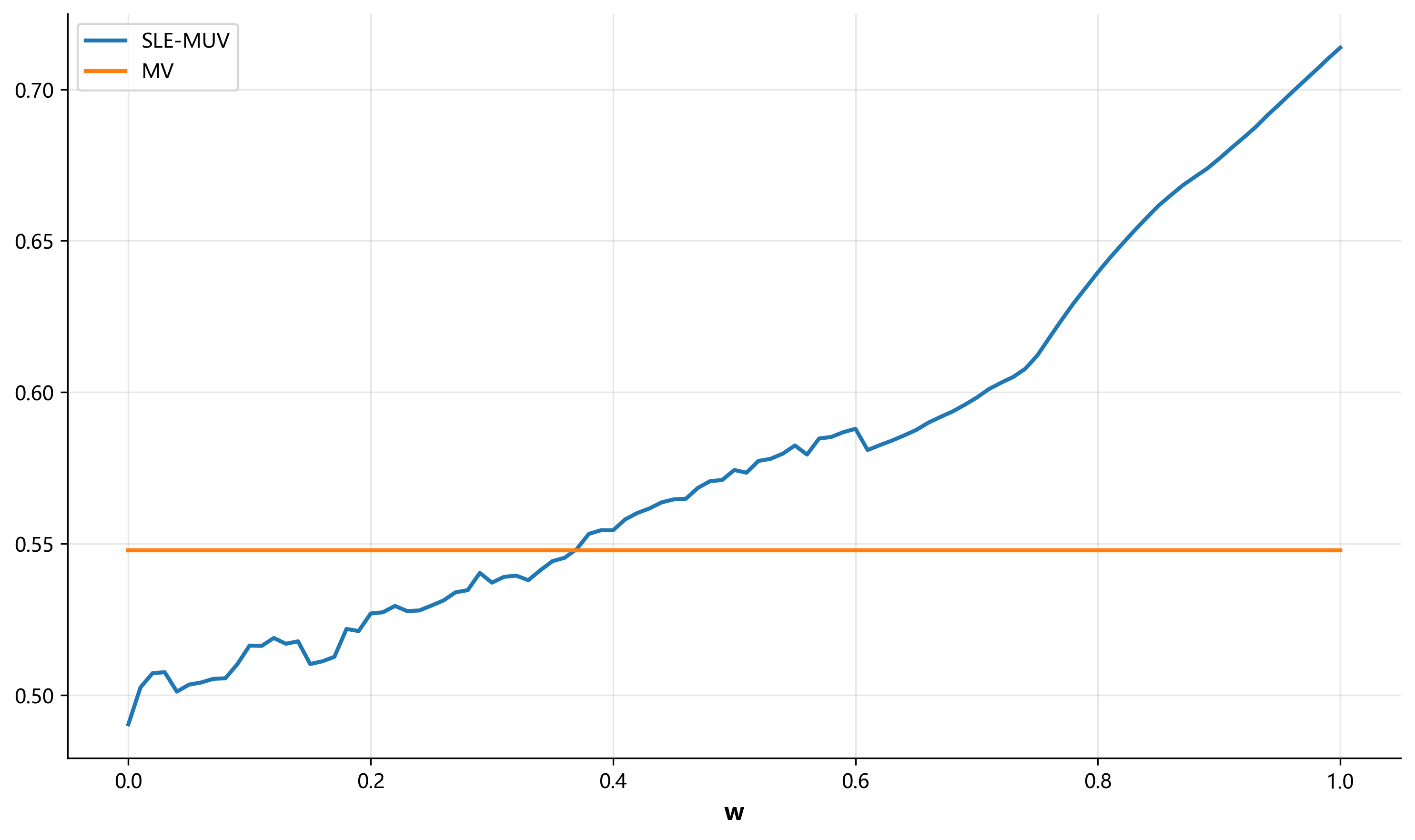}
  \caption{Cumulative Wealth at Dec 22, 2025}\label{US.6company-Accumulated Wealth Difference}
\end{subfigure}
\begin{subfigure}[htbp]{0.49\textwidth}
  \centering
  \includegraphics[width=1\textwidth]{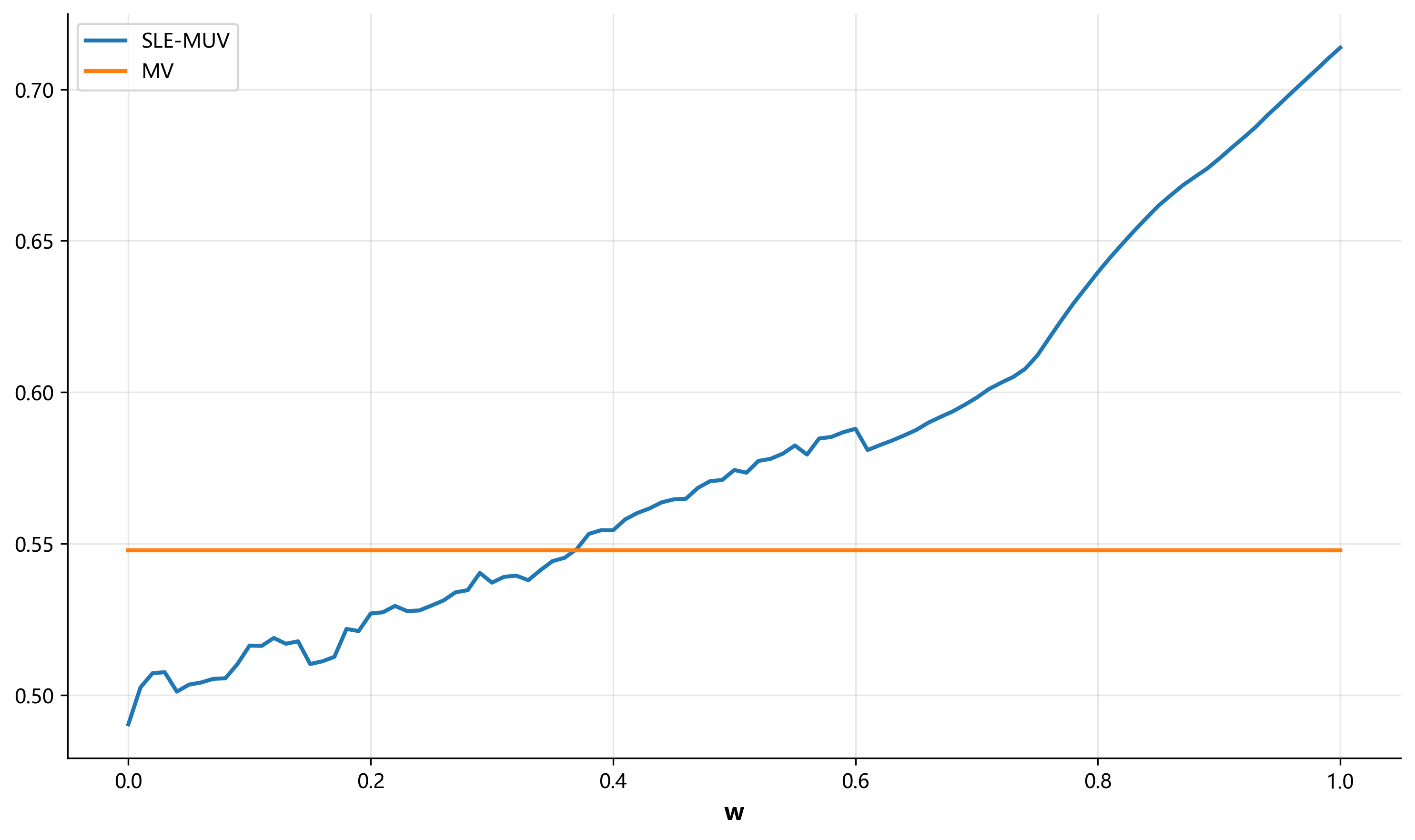}
  \caption{Sharpe Ratio}\label{US-6company-SR}
  \end{subfigure}
\caption{Cumulative Wealth as of December 22, 2025 and Sharpe Ratio : SLE-MUV Model vs. Traditional MV Model}
  \end{figure}

\begin{table}[htbp]
\caption{Cumulative Wealth (CW), Sharpe Ratio (SR), and Maximum Drawdown (MD) of the SLE-MUV Model and MV Model under Different Risk Factor w}
\label{tab:6company-CW,MD,SR}
\begin{tabular*}{\hsize}{@{}@{\extracolsep{\fill}}clllccc@{}}
\hline
 & \multicolumn{3}{c}{MUV model under SLE}                                 & \multicolumn{3}{c}{Traditional MV model}                                      \\ \hline
w    & CW    & SR    & MD     & \multicolumn{1}{l}{CW}  & \multicolumn{1}{l}{SR}  & \multicolumn{1}{l}{MD}    \\\hline
0.00                         & 1.911                         & 0.490                         & -0.291                         &                         &                         &                           \\
0.17 &1.958 & 0.513 & -0.285 &                         &                         &                           \\
0.37                         & 2.037                         & 0.548                         & -0.257                         &                         &                         &                           \\
0.5                          & 2.101                         & 0.574                         & -0.257                         &                         &                         &                           \\
1.00                         & 2.466                         & 0.714                         & -0.227                         & \multirow{-5}{*}{1.962} & \multirow{-5}{*}{0.548} & \multirow{-5}{*}{-0.2571}\\ \hline
\end{tabular*}
\end{table}
\begin{figure}[H]
  \begin{subfigure}[b]{0.49\textwidth}
  \includegraphics[width=1\textwidth]{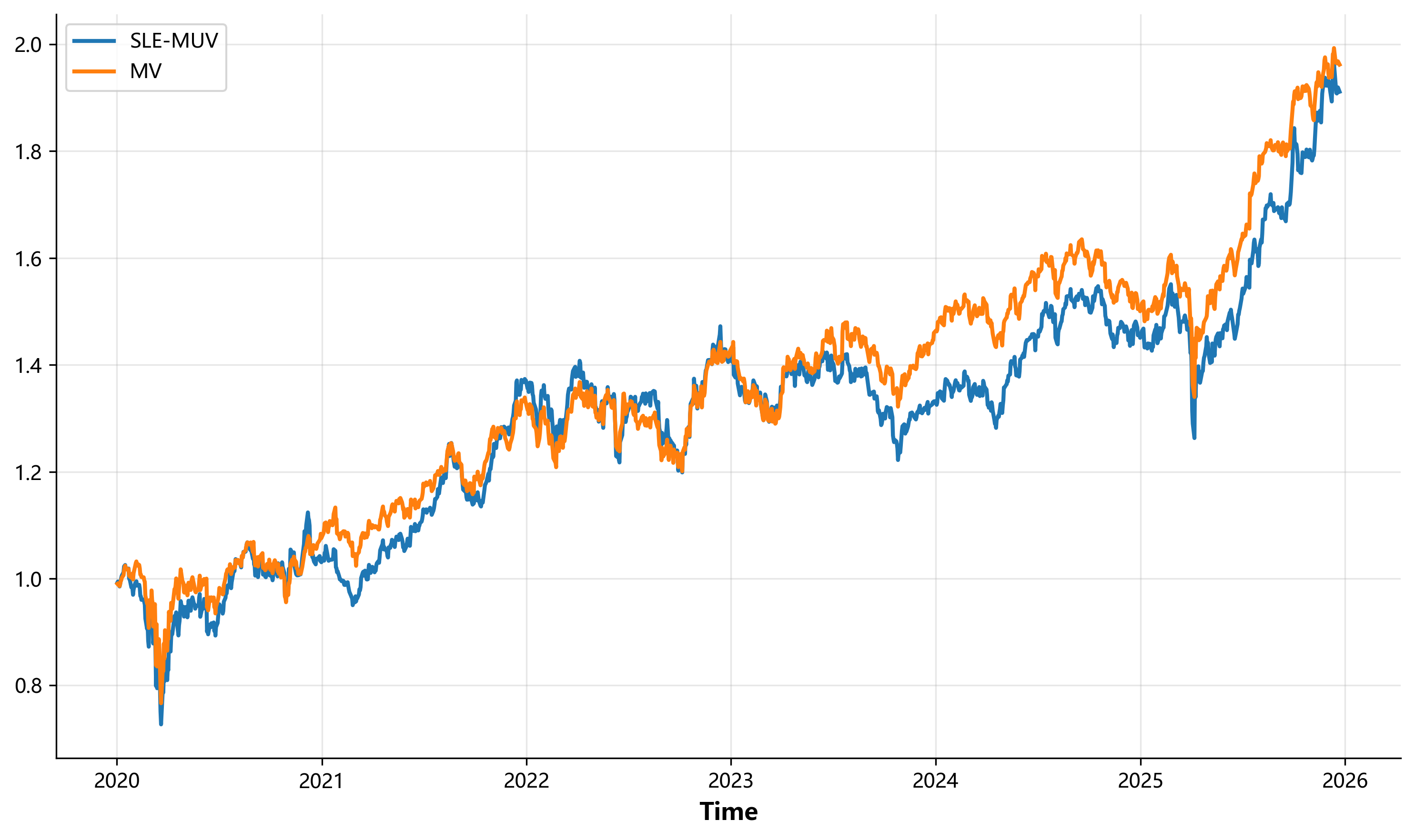}
  \caption{w=0}\label{6company w=0.}
\end{subfigure}
\begin{subfigure}[b]{0.49\textwidth}
  \centering
  \includegraphics[width=1\textwidth]{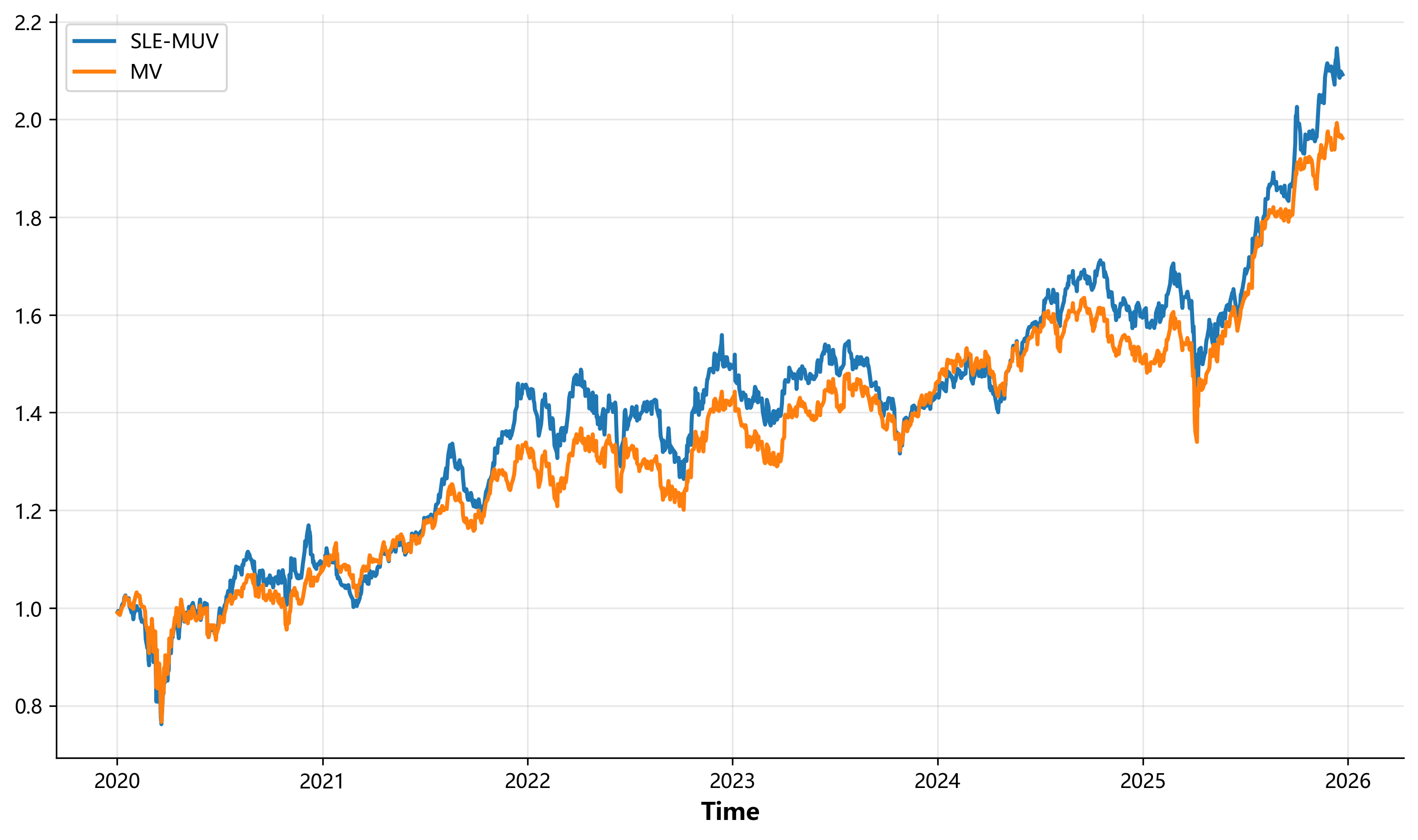}
  \caption{w=0.48.}\label{6company w=0.48}
\end{subfigure}
\begin{subfigure}[b]{0.49\textwidth}
  \centering
  \includegraphics[width=1\textwidth]{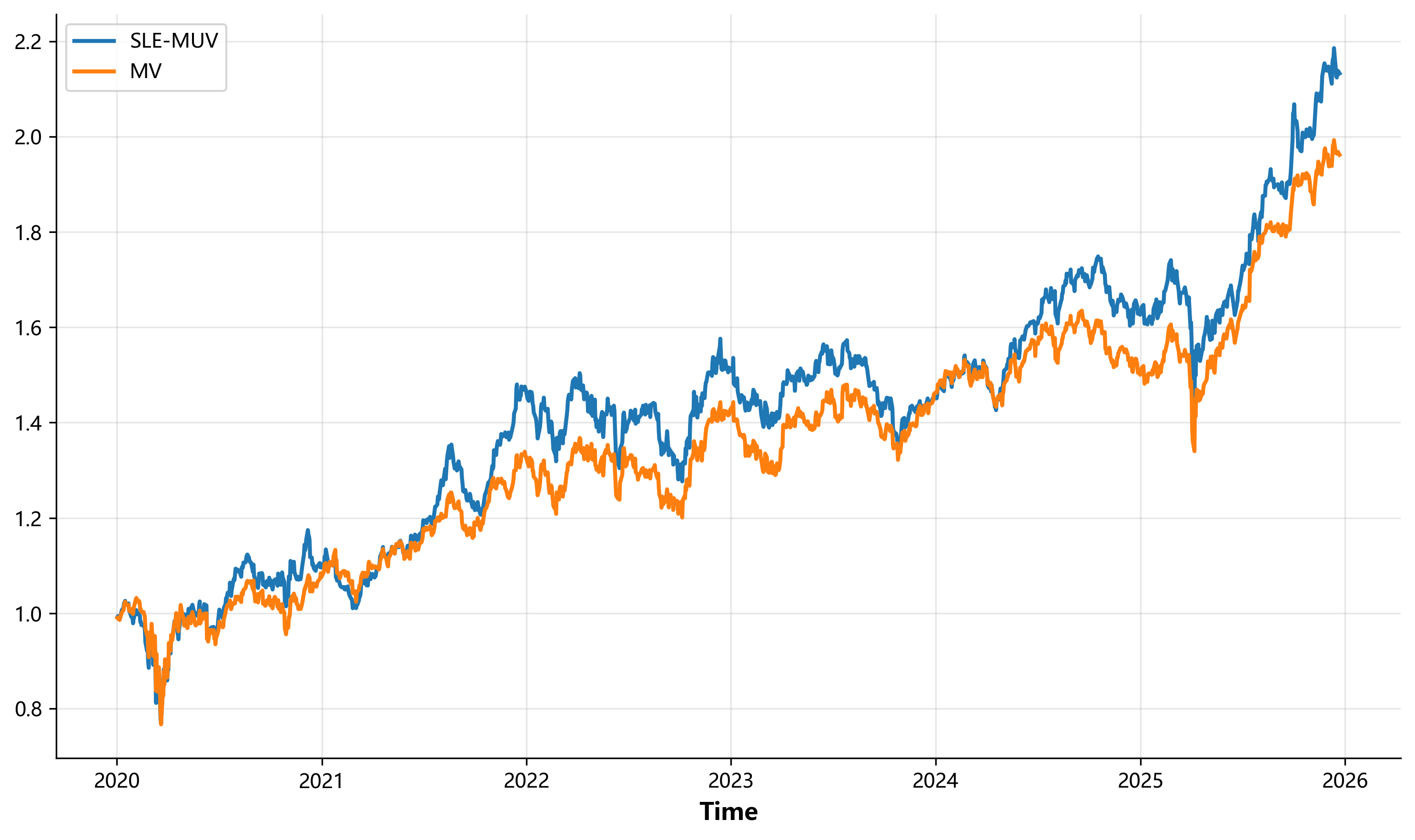}
  \caption{w=0.6.}\label{6company w=0.6}
  \end{subfigure}
\begin{subfigure}[b]{0.49\textwidth}
  \centering
  \includegraphics[width=1\textwidth]{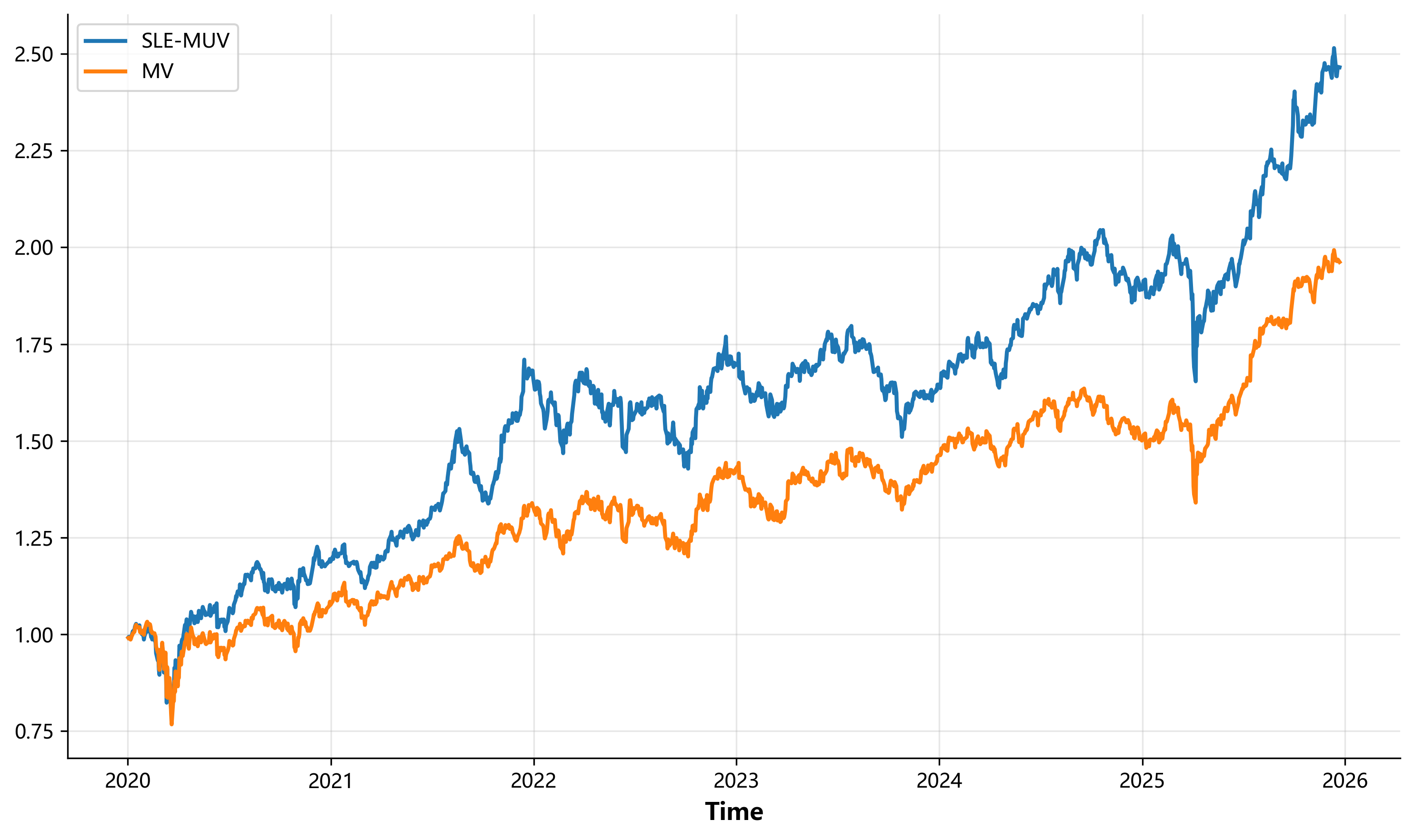}
  \caption{w=1.}\label{6company w=1}
  \end{subfigure}
  \caption{Protfolio Cumulative Wealth for SLE-MUV model and traditional MV Model under Different w Values.}\label{fig:6company data wealth fig}
\end{figure}

The difference in maximum drawdown between the SLE-MUV model and the traditional MV model lies within the interval $[-0.0338, 0.0302]$. The SLE-MUV model demonstrates favorable dynamic risk adjustment capability: the maximum drawdown narrows continuously from 29.09\% to 22.69\% as w increases, indicating a gradual optimization of risk levels. Table \ref{tab:6company-CW,MD,SR} reports the cumulative wealth, Sharpe ratio, and maximum drawdown of the SLE‑MV model and the traditional MV model under different risk factor weights w. Specifically, the values of the risk factor $w$ are selected to include two extreme cases (i.e., $w=0$ and $w=1$), as well as the value of $w$ that makes the cumulative wealth, Sharpe ratio, and maximum drawdown of the SLE-MUV model as close as possible to those of the LE model.
Figure \ref{US.6company-Accumulated Wealth Difference} and \ref{US-6company-SR} illustrate the cumulative wealth as of December 22, 2025 and Sharpe Ratio of the SLE-MUV model and the traditional MV model across different values of $w$ as of December 22, 2025, respectively. 

First, in terms of cumulative wealth  (see Figure \ref{US.6company-Accumulated Wealth Difference}), we can clearly observe that the cumulative wealth curve of the SLE-MUV model evolves from being slightly inferior to the traditional MV model in the initial phase (\( w=0 \)) to demonstrating significant excess return advantages in the later phase (\( w=1 \)). As of December 22, 2025, the cumulative wealth of the SLE-MUV model exhibits a fluctuating upward trend with the increase in the risk factor w, rising steadily from an initial value of 1.91 to 2.47. Within the intervals $w \in [0,0.09]$ and $w \in [0.15,0.17]$, the cumulative wealth of the SLE-MUV model is slightly lower than that of the traditional MV model (1.96). As w further increases, the cumulative wealth of the SLE-MUV model begins to surpass that of the traditional MV model, and this advantage continues to expand. At $w = 1$, its cumulative wealth is 0.50 higher than that of the traditional MV model.

Second, in terms of Sharpe ratio comparison (see Figure\ref{US-6company-SR}), the traditional MV model maintains a stable Sharpe ratio of 0.548. In contrast, the Sharpe ratio of the SLE-MUV model exhibits a continuous upward trend as w increases, rising steadily from an initial value of 0.490 to approximately 0.714 at $w=1$. Around $w=0.37$, the Sharpe ratio of the SLE-MUV model surpasses that of the traditional MV model for the first time; thereafter, as w further increases, the advantage of the SLE-MUV model in Sharpe ratio continues to expand. This indicates that at higher values of w, the SLE-MUV model can achieve superior risk-adjusted returns, generating higher excess returns per unit of risk.

Figure \ref{fig:6company data wealth fig} presents the cumulative wealth curve of the portfolio consisting of these six stocks over the backtesting period. We select the values of $w$ corresponding to the minimum, maximum, and near-zero differences in maximum drawdown between the SLE-MUV model and the traditional MV model, i.e., $w=0, 1, 0.48, 0.6$. As shown in Figure \ref{fig:6company data wealth fig}, during the sample period from January 1, 2020 to December 22, 2025, the cumulative wealth curves of the two models exhibit a highly similar overall fluctuation trend. However, as $w$ gradually increases, the advantage of the SLE-MUV model in cumulative wealth becomes increasingly evident, which is consistent with the findings presented in Figure \ref{US.6company-Accumulated Wealth Difference}. 

We calculate the average turnover rate and standard deviation of the SLE-MUV model and the traditional MV model when the risk factor $w=1$. The results show that the average turnover rate of the SLE-MUV model is 1.5410\% with a standard deviation of 0.0312, while the average turnover rate of the traditional MV model is 1.8204\% with a standard deviation of 0.0324. Compared with the traditional MV model, the proposed SLE-MUV model reduces the average turnover rate by 15.348\% relatively, which effectively reduces the loss of transaction costs. Meanwhile, the relative improvement in maximum drawdown of the SLE-MUV model reaches 11.7\%, showing a significant enhancement in tail risk control, and the relative increase in cumulative wealth is as high as 25.688\%. The above results illustrate that the SLE-MUV model not only features lower turnover and more robust trading behavior, but also achieves more considerable excess returns and terminal wealth accumulation while effectively controlling downside risks. It can provide investors with a flexible portfolio strategy with more controllable risks, superior returns, and stronger practical operability.

The continuous variation of the risk factor $w$  intuitively illustrates the evolutionary process of the SLE-MUV model's advantages: it is not a disruptive replacement for the traditional MV model, but rather, by introducing an additional dimension of risk adjustment (regulated by $w $), it achieves a smooth transition from ``being comparable to the benchmark'' to ``significantly outperforming the benchmark''. 
This progressive improvement in performance not only reflects the rationality of the model design but also provides a clear practical basis for dynamically adjusting $w$ according to risk preferences in actual investments.
Overall, investors can select an appropriate w based on their investment objectives and the current market environment, enabling the SLE-MUV model to achieve higher cumulative returns while effectively controlling maximum drawdown, thereby attaining a more favorable risk-return trade-off. This result fully validates the robustness and effectiveness of the model in practical investment applications.
\subsubsection{A-shares stocks}
In this section, we selected the return data of six A-share stocks from January 2, 2019 to December 25, 2025  for empirical analysis.
The sample includes: BYD, Industrial and Commercial Bank of China (ICBC), Kweichow Moutai, Hengrui Medicine, Sany Heavy Industry, and Wuliangye. 
It is observed that the cumulative wealth of the SLE-MUV model is consistently higher than that of the traditional MV model for any $w \in [0,1]$. Accordingly, investors can flexibly select the risk factor to maximize cumulative wealth (e.g., setting $w=1$). Furthermore, the proposed SLE-MUV model delivers a higher Sharpe ratio relative to the MV model. In terms of downside risk, the absolute difference in maximum drawdown between the two models is limited to 1.51\%, corresponding to a relative difference of 7.29\%.

Figure \ref{Ashares6company-CW SR} depicts the cumulative wealth curves of the SLE-MUV model and the MV model as of December 22, 2025, as well as the trend chart of the Sharpe ratio varying with the risk factor $w$, respectively.
For any risk weight w in the interval [0,1], the cumulative wealth of the SLE-MUV model is consistently higher than that of the traditional MV model, and reaches its maximum when $w=1$. The trend of the Sharpe ratio of the SLE-MUV model is similar to that of its cumulative wealth curve. When $w \in [0,0.04] \cup [0.08,0.09] \cup [0.18,1]$, the Sharpe ratio of the SLE-MUV model significantly surpasses that of the traditional MV model, indicating superior risk-adjusted return efficiency. 
Regarding tail risk, although the absolute maximum drawdown of the SLE-MUV model remains higher than that of the traditional MV model, and attains its maximum at $w=0.18$. As $w$ increases, the gap between the SLE-MUV model and the MV model gradually narrows, with a relative difference of 7.29\% at $w=1$.
This indicates that the risk controllability of the model is significantly improved over the long-term sample period.
Table \ref{tab:Acompany6shares} summarizes the cumulative wealth, Sharpe ratio and maximum drawdown of the SLE-MUV and conventional MV models. The selected values of the risk parameter w cover the two endpoints of the interval, i.e., $w=0,1$. In addition, four critical parameter values are considered: two that minimizes the Sharpe ratio difference between the two models, and others that yield the global maximum drawdown of the SLE-MV model across all feasible w values (w=0.07, 0.14, 0.17 and 0.18).

\begin{figure}[H]
  \centering
  \begin{subfigure}[htbp]{0.49\textwidth}
  \centering
  \includegraphics[width=1\textwidth]{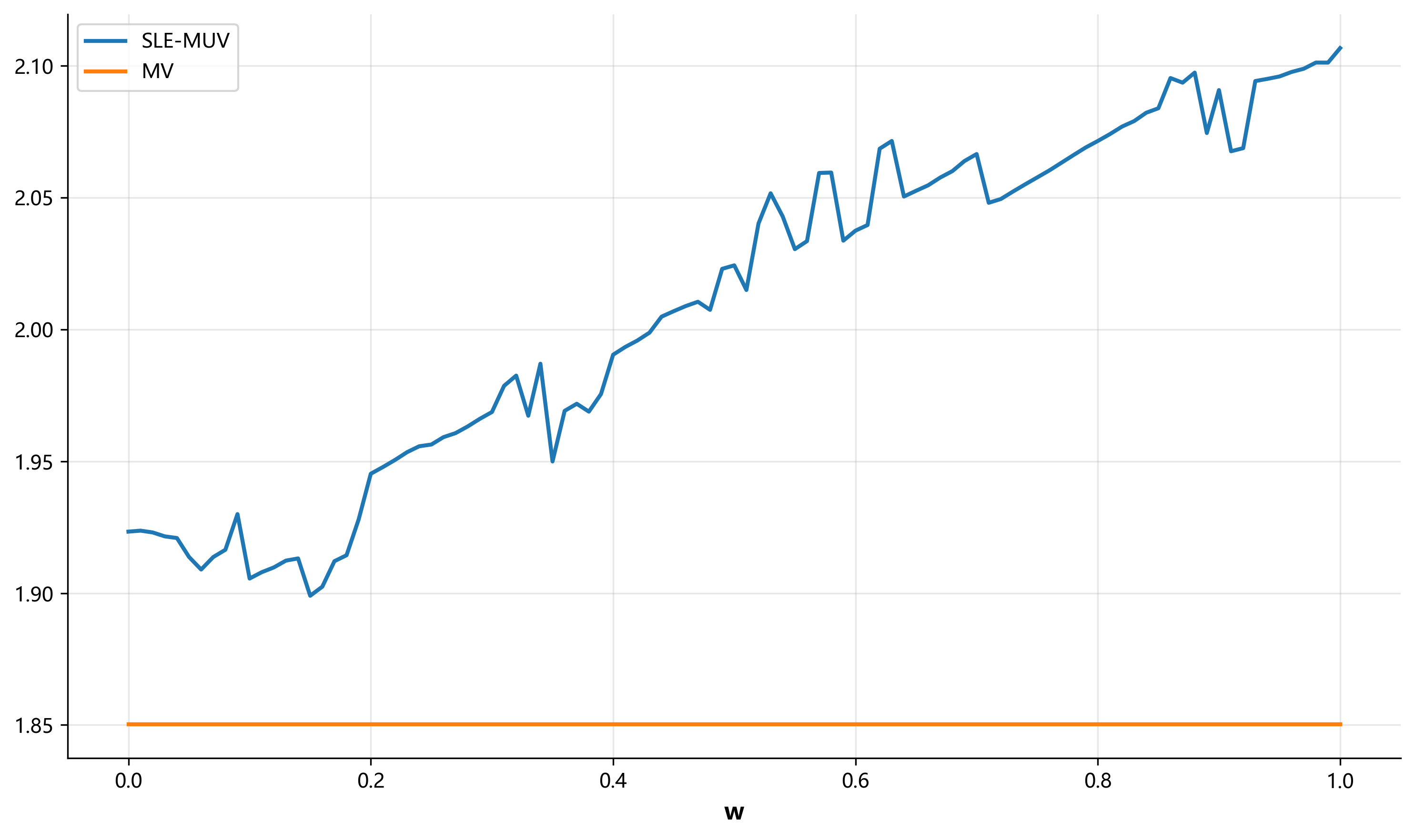}
  \caption{Cumulative Wealth at Dec 22, 2025}
\end{subfigure}
\begin{subfigure}[htbp]{0.49\textwidth}
  \centering
  \includegraphics[width=1\textwidth]{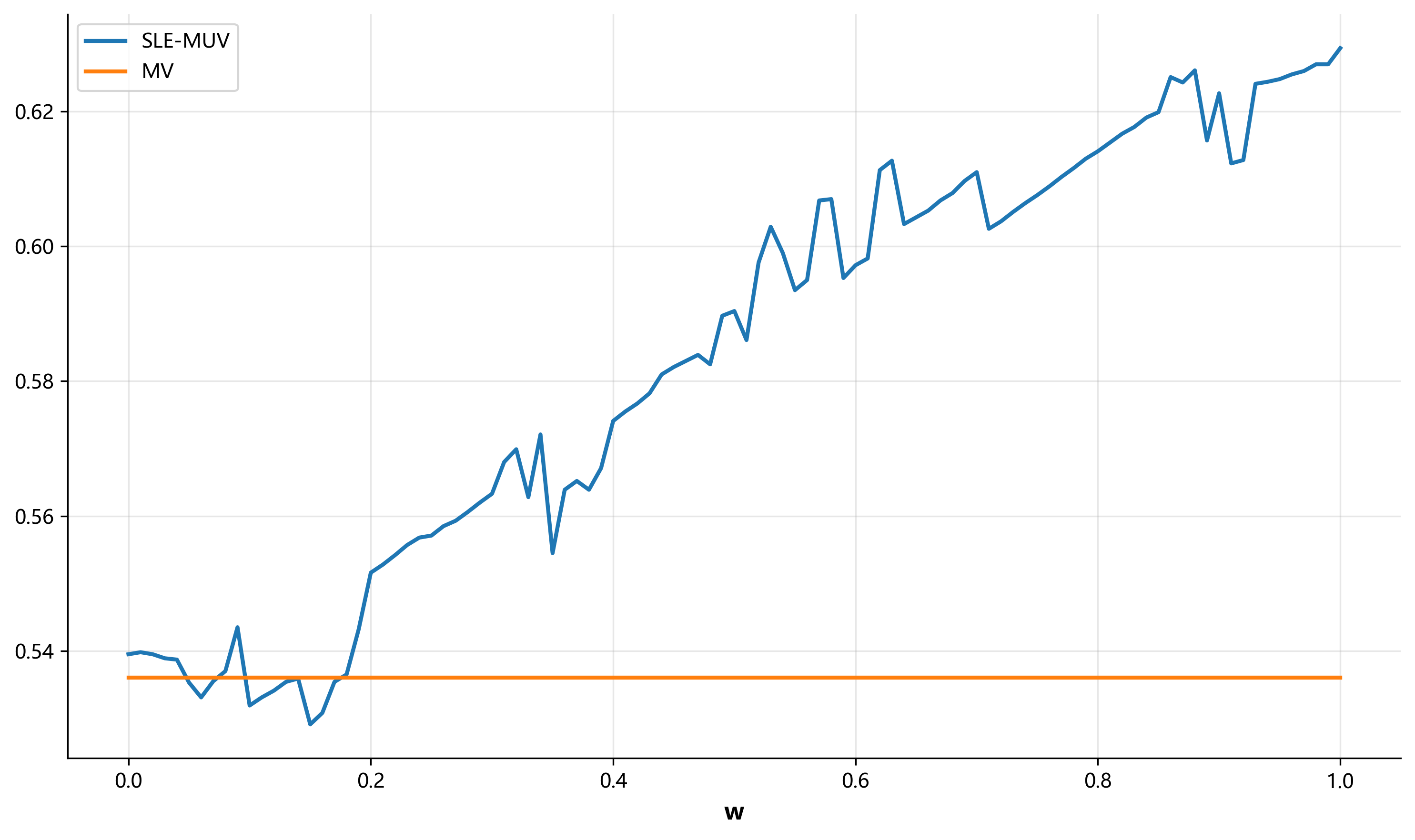}
  \caption{Sharpe Ratio}
  \end{subfigure}
\caption{Cumulative Wealth as of December 22, 2025, Sharpe Ratio: SLE-MUV Model vs. Traditional MV Model}
\label{Ashares6company-CW SR}
  \end{figure}

\begin{table}[H]
\caption{Cumulative Wealth (CW), Sharpe Ratio (SR), and Maximum Drawdown (MD) of the SLE-MUV Model and MV Model under Different Risk Factor w.}
\label{tab:Acompany6shares}
\begin{tabular*}{\hsize}{@{}@{\extracolsep{\fill}}clllccc@{}}
\hline
 & \multicolumn{3}{c}{MUV model under SLE}                                 & \multicolumn{3}{c}{Traditional MV model}                                     \\\hline
w    & CW    & SR    & MD     & \multicolumn{1}{l}{CW}  & \multicolumn{1}{l}{SR}  & \multicolumn{1}{l}{MD}   \\\hline
0.00                         & 1.923                         & 0.540                         & -0.249                         &                         &                         &                          \\
0.07 & 1.914 & 0.536 & -0.251 &                         &                         &                          \\
0.14                         & 1.913                         & 0.536                         & -0.252                         &                         &                         &                          \\
0.17                         & 1.912                         & 0.535                         & -0.253                         &                         &                         &                          \\
0.18                         & 1.914                         & 0.537                         & -0.253 &                         &                         &                          \\
1                            & 2.107                         & 0.629                         & -0.222                         & \multirow{-7}{*}{1.850} & \multirow{-7}{*}{0.536} & \multirow{-7}{*}{-0.207}\\ \hline
\end{tabular*}
\end{table}

\begin{figure}[H]
  \centering
  \begin{subfigure}[htbp]{0.49\textwidth}
  \includegraphics[width=1\textwidth]{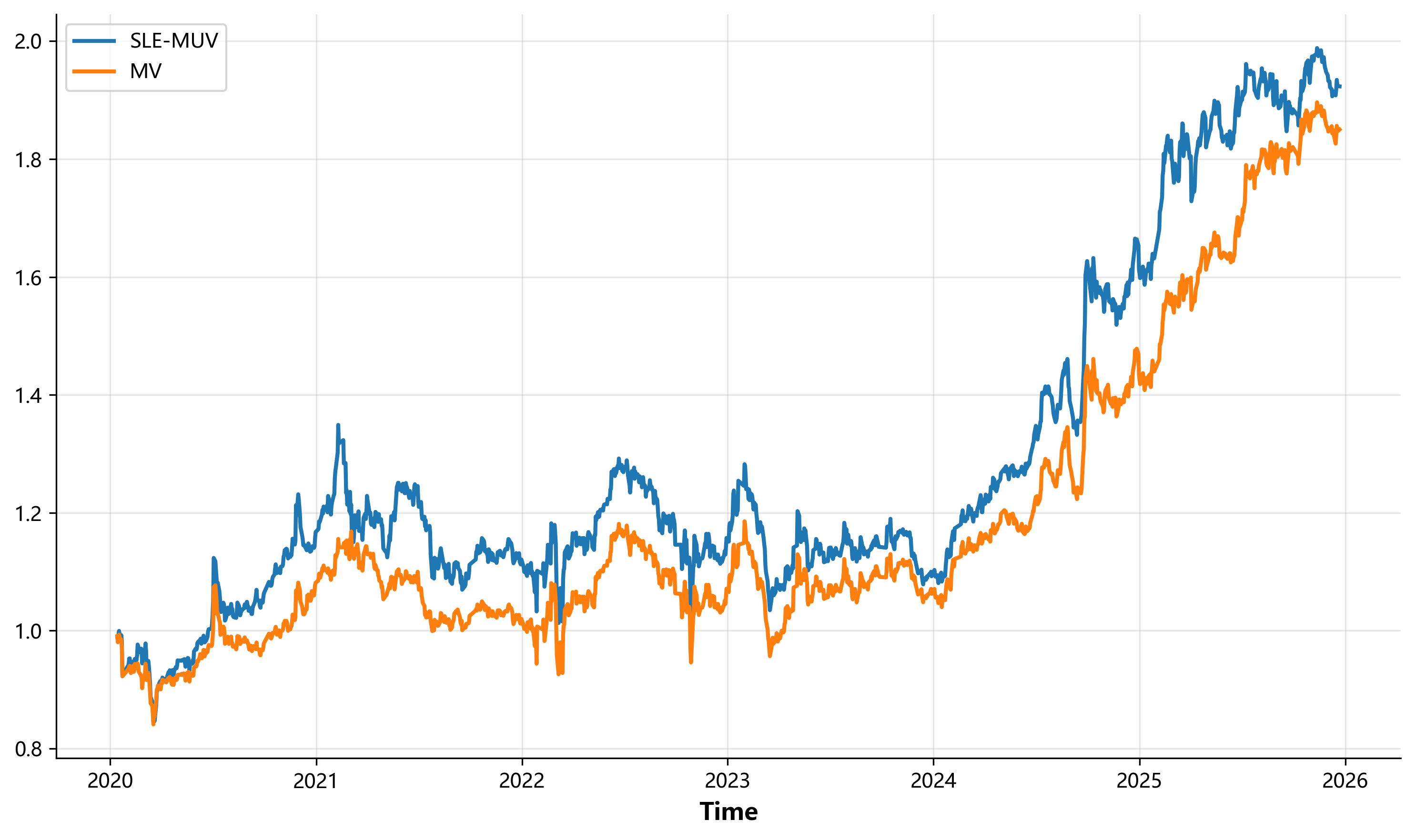}
  \caption{w=0.}\label{Ashares_6company w=0}
\end{subfigure}
\begin{subfigure}[htbp]{0.49\textwidth}
  \centering
  \includegraphics[width=1\textwidth]{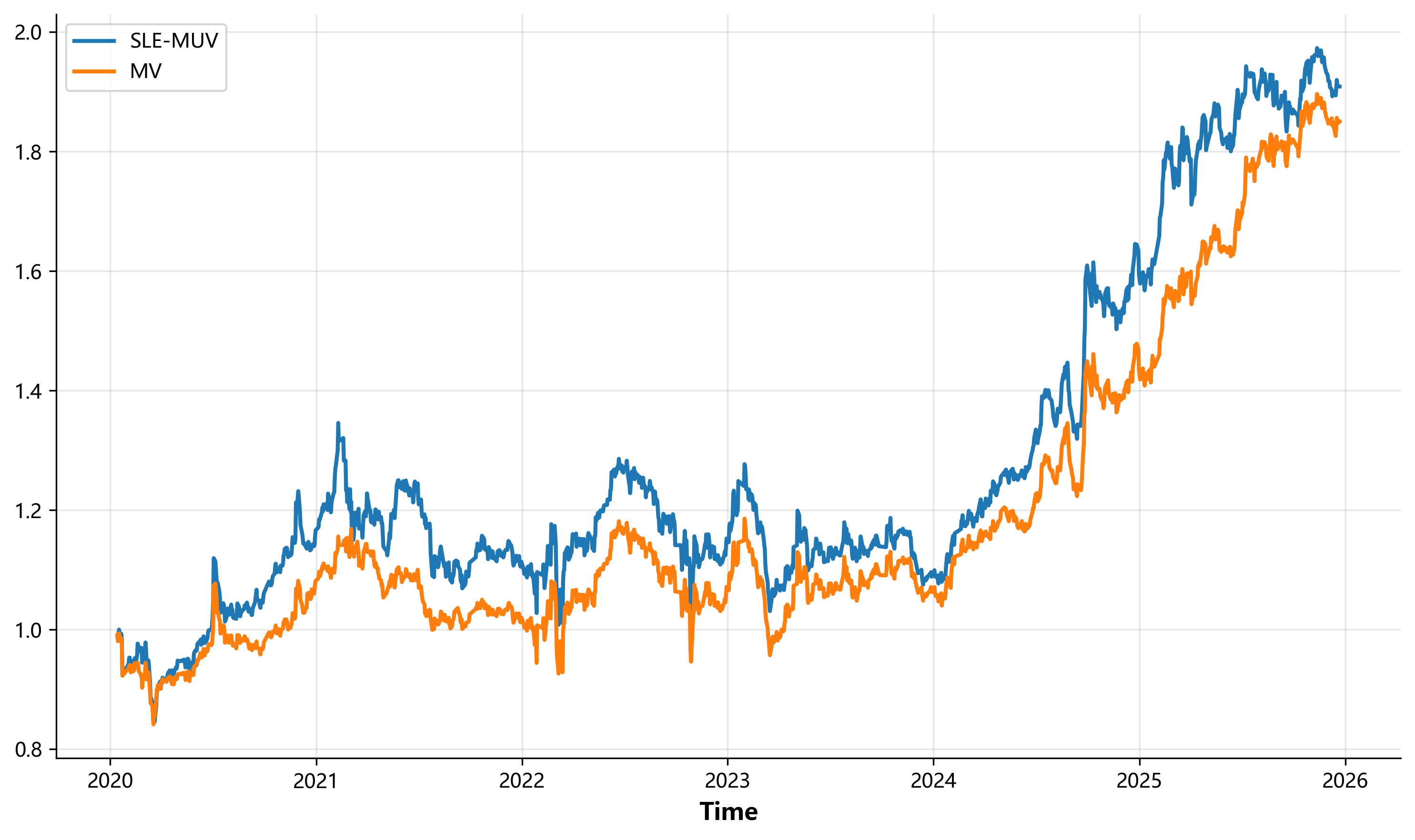}
  \caption{w=0.06.}\label{Ashares_6company w=0.06}
\end{subfigure}
\begin{subfigure}[htbp]{0.49\textwidth}
  \centering
  \includegraphics[width=1\textwidth]{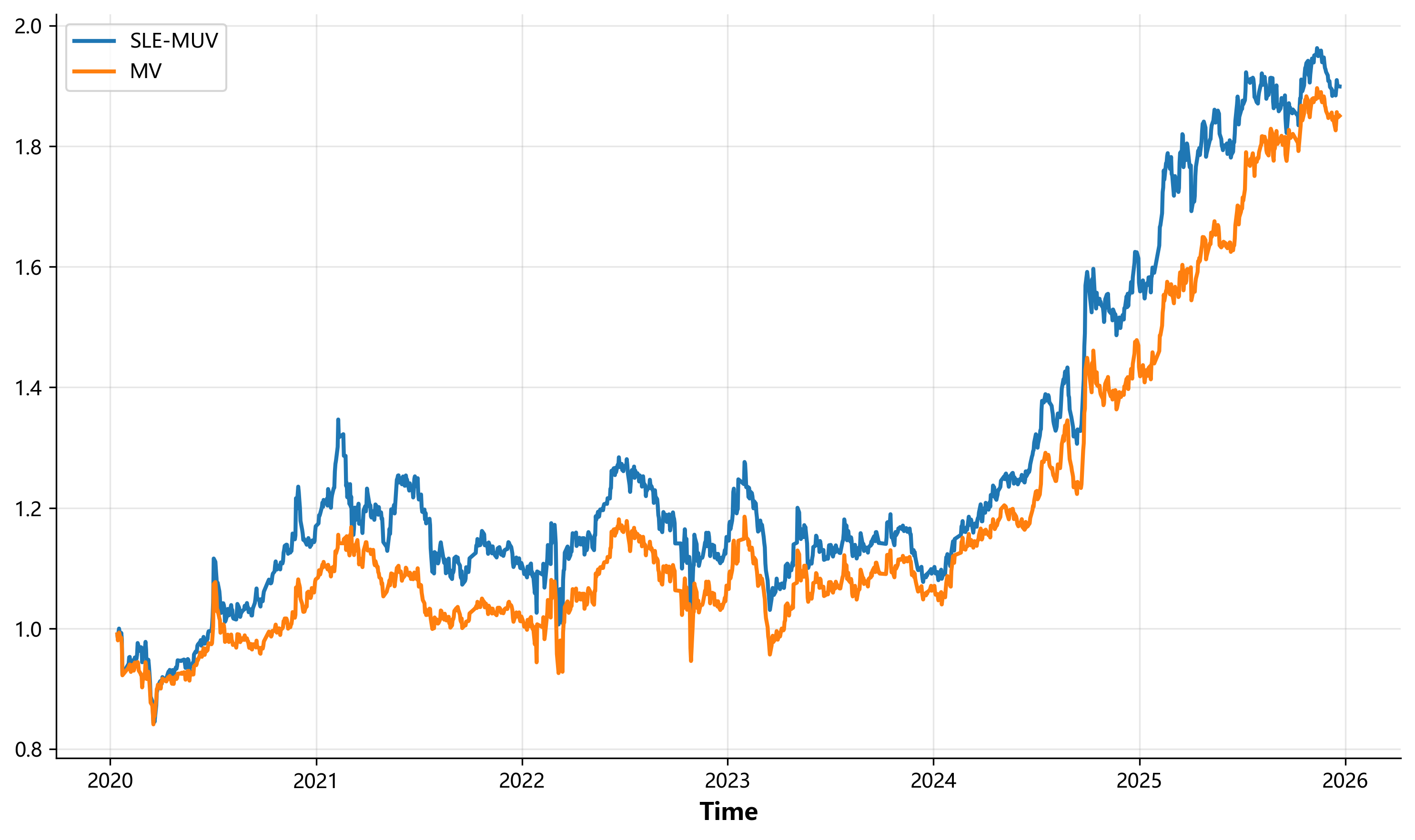}
  \caption{w=0.15.}\label{Ashares_6company w=0.15}
  \end{subfigure}
\begin{subfigure}[htbp]{0.49\textwidth}
  \centering
  \includegraphics[width=1\textwidth]{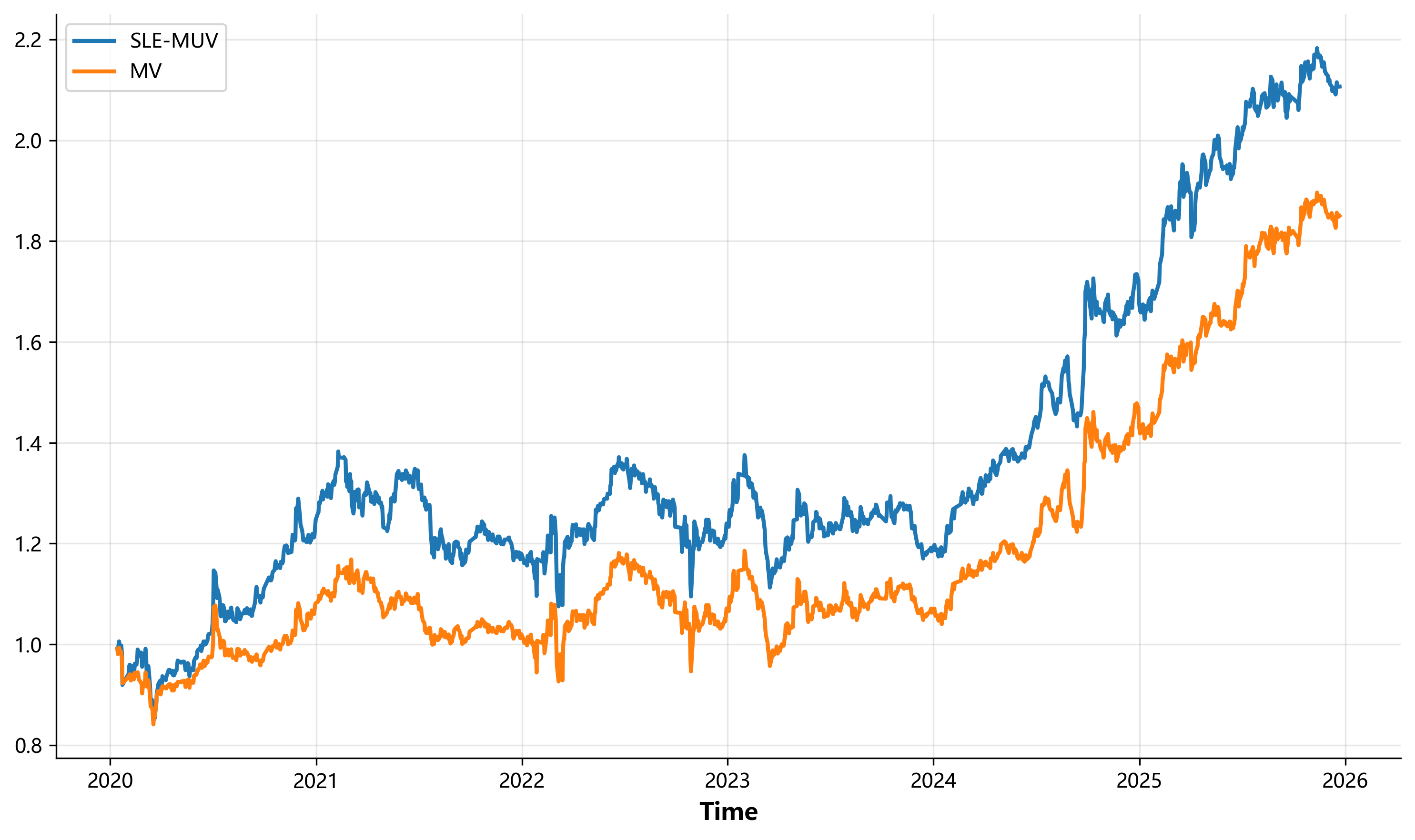}
  \caption{w=1.}\label{Ashares_6company w=1}
  \end{subfigure}
  \caption{Portfolio Cumulative Wealth for SLE-MUV model and traditional MV Model under under Different Values of w.}\label{fig:Ashares_6company data wealth fig}
\end{figure}

We further select two local minima where the Sharpe ratio of the SLE-MUV model falls below that of the traditional MV model ($w = 0.06, 0.15$), along with w=0 and w=1, for comparative analysis. The results align with the conclusions in Subsections \ref{sec: Synthetic Data} and \ref{sec:US stocks}: during the sample period from January 1, 2020, to December 22, 2025, the cumulative wealth curve of the SLE-MUV model shares a highly similar fluctuation pattern with the traditional MV model. However, by dynamically adjusting the parameter $w$, the SLE-MUV model enables higher long-term returns while maintaining manageable risk exposure.
Investors with strong risk tolerance, a long investment horizon and no pursuit of short-term market fluctuations can choose an appropriate risk factor to achieve their preferred investment allocation.

Experimental results from both synthetic data and real markets (U.S. stocks and China's A-shares) show that the performance advantage of the SLE-MUV model over the traditional MV model exhibits a non-monotonic pattern as the weight parameter $w$ varies from 0 to 1. This behavior arises because the optimal portfolio weights themselves vary non-monotonically with $w \in [0,1]$, as established in Subsection \ref{sec:The role of weights in portfolio allocation}.
Accordingly, at each rebalancing date, we can determine an optimal weight $w^* \in [0,1]$ that maximizes the portfolio return and minimizes the variance simultaneously. In this way, we obtain an optimal portfolio strategy that achieves the maximal cumulative wealth.

\section{Conclusion}\label{sec:conclusion}
This paper incorporates sublinear expectation theory into the classical mean-variance (MV) framework to derive the upper and lower variance bounds of a portfolio. We introduce a risk factor $w \in [0,1]$, which explicitly accounts for both the upper and lower variance bounds of the portfolio, to transform the multi-objective optimization problem into a tractable single-objective formulation, namely the SLE-MUV model. 
This framework enables more precise risk control in portfolio management practice.
We prove that the Pareto frontier of the proposed SLE-MUV model is a continuous convex curve, and derive the closed-form analytical expression of the optimal solution to the SLE-MUV model by the active set method.
Under SLE-MUV model, investors can select the appropriate value of the risk factor w to implement portfolio allocation according to their individual investment preferences.
The optimal value of the risk factor $w$ varies with the selection of individual stocks included in the portfolio.

To evaluate the performance of the proposed SLE-MUV model and benchmark it against the classical MV model, we conduct comprehensive empirical analyses using three distinct datasets: numerically simulated synthetic data, a sample of six representative U.S. stocks, and a sample of six representative  A-shares.
The empirical results show that investors can flexibly adjust the risk factor w to achieve personalized investment objectives, conditional on prevailing market conditions and their idiosyncratic risk preferences. These comprehensive empirical tests confirm the effectiveness, flexibility, and practical applicability of our proposed method.

For future research, we will extend the SLE-MUV model to the multi-period portfolio optimization setting, incorporate multi-period asset allocation and dynamic trading decisions, and optimize the overall performance of long-horizon portfolios based on the sublinear expectation framework.
\bibliography{references}
\section*{Appendix}
\subsection*{A. Sublinear expectation}
\label{subsec:Appendix_A}
Suppose that $\Omega$ is a given set, and let $\mathcal{H}$ denote a linear space of real-valued functions which are defined on $\Omega$. $C_{b,ip}(\mathbb{R}^n)$ denotes the space of bounded and Lipschitz continuous functions. $C_{l,Lip}(\mathbb{R}^n)$ represents the linear space of functions $\varphi$ satisfying the following local Lipschitz condition:
\begin{align*}
	|\varphi(x)-\varphi(y)| &\le C(1+ |x|^m + |y|^m)|x-y|,\  \text{for}\  x,y \in \mathbb{R}^n, \\
	&\text{where the constant $C>0$ and the integer $m \in \mathbb{N}$ depend on $\varphi$ }.
\end{align*}

\begin{definition}[\cite{Peng2010}]
	A functional $\mathbb{E}: \mathcal{H}\to \mathbb{R}$ satisfying
	\begin{enumerate}
		\item[(i)] $\mathbb{E}[X]\le \mathbb{E}[Y], \ \forall  X \le Y$.
		\item[(ii)] $\mathbb{E}[c] = c,\ \forall c\in \mathbb{R}$ .
		\item[(iii)] $\mathbb{E}[X+Y] \le \mathbb{E}[X]+\mathbb{E}[Y], \ \forall X, Y\in \mathcal{H}$.
		\item [(iv)] $\mathbb{E}[\lambda X] = \lambda \mathbb{E}[X]$,\  $\forall \lambda \ge 0$.
	\end{enumerate}
	
	Then $\mathbb{E}$ is called a \textbf{Sublinear Expectation}, and the triplet $(\Omega,\mathcal{H},\mathbb{E})$ is called a \textbf{Sublinear Expectation Space}.
\end{definition}
\begin{definition}[\cite{Peng2010}]
	Let $X_1$ and $X_2$ be two n-dimensional random vectors defined on nonlinear expectation spaces $(\Omega_1,\mathcal{H}_1,\mathbb{E}_1)$ and $(\Omega_2,\mathcal{H}_2,\mathbb{E}_2)$, respectively. They are called \textbf{identically distributed}, denoted by $X_1 \overset{d}{=} X_2$, if
	\begin{align*}
		\mathbb{E}_1[\varphi(X_1)] = \mathbb{E}_2[\varphi(X_2)], \ \forall \varphi \in C_{b,lip}(\mathbb{R}^n).
	\end{align*}
\end{definition}
\begin{definition}[\cite{Peng2010}]
	In a nonlinear expectation space $(\Omega,\mathcal{H},\mathbb{E})$, a random vector $Y \in \mathcal{H}^d$ is said to be \textbf{independent} of another random vector $X \in \mathcal{H}^m$ under $\mathbb{E}$ if for each test function $\varphi \in C_{b,lip}(\mathbb{R}^{m+n})$, we have
	\begin{align*}
		\mathbb{E}[\varphi(X,Y)] = \mathbb{E}[\mathbb{E}[\varphi(x,Y)]_{x=X}].
	\end{align*}
\end{definition}
\begin{definition}[\cite{Peng2010}]
	A $d$-dimensional random vector $\eta = (\eta_1,\eta_2,...,\eta_d)$ on a sublinear expectation space $(\Omega,\mathcal{H},\mathbb{E})$ is called \textbf{maximally distributed} if there exists a bounded, closed and convex subset $\Gamma \subset \mathbb{R}^d$ such that
	\begin{align*}
		\mathbb{E}[\varphi(\eta)] = \max_{y \in \Gamma}\varphi(y), \ \varphi \in C_{l,Lip}(\mathbb{R}^d) .
	\end{align*}
\end{definition}
\begin{definition}[\cite{Peng2010}]
  d-dimensional random vector $X = (X_1,X_2,...,X_d)$ on a sublinear expectation space $(\Omega,\mathcal{H},\mathbb{E})$ is called G-normally distributed if
  \begin{align}
    aX+b\overline{X} = \sqrt{a^2+b^2}X \ for \ a,b\ge 0,
  \end{align}
  where $\overline{X}$ is an independent copy of X. 
\end{definition}
\begin{theorem}[\cite{Peng2010}]\label{th:law of large numbers}
	Let $\{Y_i\}_{i=1}^{\infty}$ be a sequence of $\mathbb{R}^d$-valued random variables on  $(\Omega,\mathcal{H},\mathbb{E})$. We assume that $Y_{i+1} \overset{d}{=} Y_{i}$ and  $Y_{i+1}$ is independent from $\{Y_1,...,Y_i\}$ for each $i = 1,2,...$.We assume further the following uniform integrability condition:
	\begin{align*}
		\lim_{\lambda \to +\infty}\mathbb{E}[(|Y_1|-\lambda)^+] = 0.
	\end{align*}
	Then the sequence $\{ \frac{1}{n}(Y_1+Y_2+...+Y_n)\}_{n=1}^{\infty}$ converges in law to a maximal distribution, i.e.
	\begin{align*}
		\lim_{n\to \infty}\mathbb{E}\left [\varphi \left (\frac{1}{n}(Y_1+Y_2+...+Y_n)\right)\right ] = \max_{\theta \in \overline{\Theta}} \varphi(\theta),
	\end{align*}
	for all functions $\varphi \in C(\mathbb{R}^d)$ satisfying linear growth condition, i.e., $|\varphi(x)|\le C(1+|x|)$, where $\overline{\Theta}$ is the (unique) bounded, closed and convex subset of $\mathbb{R}^d$ satisfying
	\begin{align*}
		\max_{\theta \in \overline{\Theta}}\left \langle p,\theta \right \rangle = \mathbb{E}[\left \langle p,Y_1 \right \rangle], p \in \mathbb{R}^d.
	\end{align*}
\end{theorem}
\subsection*{B. Pareto frontier}
For a Multi-objective Optimization Problem (MOP) 
\begin{align*}
  (\text{MOP})\quad  \underset{x\in X}\min f(x) = \left(f_1(x), f_2(x), ..., f_p(x)\right)^{\top}
\end{align*}
where
\begin{align*}
  X = \left\{ x \in \mathbb{R}^n | g(x) = \left( g_1(x), g_2(x),...,g_m(x) \right)^{\top} \le 0, h(x) = \left( h_1(x),h_2(x),...,h_{l}(x) \right)^{\top} = 0 \right\}.
\end{align*}
\begin{definition}
    Let $x^*\in X$. If for any $ x \in X $, $ f(x^*) \le f(x) $ holds, i.e., $ f_i(x^*) \leq f_i(x) $ for all $ i = 1, 2, \dots, p $, then $ x^* $ is said to be an \textbf{absolute optimal solution} of MOP.
\end{definition}
\begin{definition}
  Let $ x^* \in X $. If there does not exist $ x \in X $ such that $ f(x) \leq f(x^*) $ (or $ f(x) < f(x^*) $), then $ x^* $ is said to be an \textbf{Efficient solution} (or \textbf{Weakly Efficient solution}) of MOP.
\end{definition}
  All efficient solutions, which are also known as Pareto optimal solutions, form a set called the Pareto optimal set. The corresponding objective vectors are considered to lie on the Pareto frontier.

\end{document}